\documentclass[runningheads]{llncs}
\usepackage{todonotes}
\usepackage[T1]{fontenc}
\usepackage{graphicx}
\usepackage{comment}

\usepackage{amsmath}
\usepackage{amssymb}
\usepackage{amsfonts}
\usepackage{nicematrix}
\usepackage{appendix}
\usepackage[shortlabels]{enumitem}

\newcommand{\PP}{\mathcal P}

\newcommand{\sgn}{\mathrm{sgn}}
\newcommand{\q}[1]{``#1''}
\newcommand{\vc}[1]{\mathbf{#1}}
\newcommand{\RR}{\mathbf{R}}
\newcommand{\NN}{\mathbf{N}}
\newcommand{\Rk}[1]{\mathrm{Rank}(#1)}
\newcommand{\ce}[1]{\mathfrak{#1}}
\newcommand{\cec}[1]{\mathfrak{C}^{\mathrm{#1}}_G}

\newcommand{\cel}[1]{\mathfrak{L}^{\mathrm{\mathbf{#1}}}_G}
\newcommand{\celg}[2]{\mathfrak{L}^{\mathrm{\mathbf{#1}}}_{#2}}

\newcommand{\myvdots}{\raisebox{0.34\baselineskip}{\ensuremath{\vdots}}}
\newcommand{\adj}{\,\text{\textemdash}\,}

\usepackage{ntheorem}
\newtheorem*{theorem-non}{Theorem}

\usepackage{tikz}

\usetikzlibrary{calc,decorations.pathmorphing,decorations.pathreplacing,arrows.meta,positioning,snakes,quotes,shapes.geometric,fadings,math}
\begin{document}
\title{Properties and Expressivity of\\Linear Geometric Centralities\thanks{This is an extended version of a paper presented at the 20th Workshop on Modelling and Mining Networks~\cite{BFFLGC}.}}
\author{Paolo Boldi \and%\orcidID{0000-0002-8297-6255} \and
Flavio Furia \and%\orcidID{0009-0003-5265-4161}\and
Chiara Prezioso}
\authorrunning{P. Boldi et al.}

\institute{Dipartimento di Informatica, Universit\`a degli Studi di Milano, Italy
            \\ \email{paolo.boldi@unimi.it}}
\maketitle              % typeset the header of the contribution
\begin{abstract}
Centrality indices are used to rank the nodes of a graph by importance: this is a common need in many concrete situations (social networks, citation networks, web graphs, for instance) and it was discussed many times in sociology, psychology, mathematics and computer science, giving rise to a whole zoo of definitions of centrality.
Although they differ widely in nature, many centrality measures are based on shortest-path distances: such centralities are often referred to as \emph{geometric}. Geometric centralities can use the shortest-path-length information in many different ways, but most of the existing geometric centralities can be defined as a linear transformation of the distance-count vector (that is, the vector containing, for every index $t$, the number of nodes at distance $t$).

In this paper we study this class of centralities, that we call \emph{linear (geometric) centralities}, in their full generality. In particular, we look at them in the light of the axiomatic approach, and we study their expressivity: we show to what extent linear centralities can be used to distinguish between nodes in a graph, and how many different rankings of nodes can be induced by linear centralities on a given graph. The latter problem (which has a number of possible applications, especially in an adversarial setting) is solved by means of a linear programming formulation, which is based on Farkas' lemma, and is interesting in its own right.
\end{abstract}
\section{Introduction and Motivation}
\label{sec:intro}

Network science is a discipline focused on analyzing network structures in order to obtain information that can be used to understand the underlying domain of the network itself.
One key problem that network scientists often face is to determine the importance of nodes in a network: this is crucial for understanding the dynamics of the system and for making informed decisions based on the network's structure. Of course, the definition of importance is not unique, and it can vary depending on the context and the specific properties of the network being analyzed. 
Centralities have been defined and re-defined many times over the years, in different contexts and with different approaches, and they are used to rank nodes in a network by their importance~\cite{BoVAC}.

In this work, we study indices that are often referred to as \emph{geometric}, as they solely depend on how many nodes exist at every (shortest-path) distance.
Notable examples of geometric centralities include in-degree, closeness~\cite{BavMMGS,BavCPTOG}, and harmonic~\cite{BoVAC} centrality.
In particular, we propose and study a wide class of geometric measures that can be computed as linear combination of the number of nodes at every distance: these centralities will be referred to as \emph{linear geometric centralities} (or just linear centralities, for short).
The vast majority of the most used geometric centralities are linear (in particular, all centralities mentioned above are equivalent to some linear centrality).

A special case of linear centralities are those assigning larger weights to shorter distances: the intuition here is that nearby connections are more relevant than long-range ones.
These types of decaying centralities first appeared in~\cite{KisOCFG} (albeit with the equivalent but opposite postulation that smaller centrality values denote a higher importance). In~\cite{KiTTCFDG} the same class was extended to directed graphs, and several of its properties were explored.

In more recent years, decaying centralities appeared in~\cite{SkSADBC}, where they are called additive and studied through the so-called ``axiomatic approach'' (pioneered in~\cite{SabCIG}).
Another example of this approach used to study these centralities can be found in~\cite{GarAFCN} and~\cite{BDFSRSMCBDDC}.
In the latter, the impact of edge additions to the final scores and rankings yielded by these indices was investigated (obtaining, independently, results similar to those in~\cite{KisOCFG}).

In the present paper, we relax the assumption of decaying weights, and study linear geometric centralities in their full generality: our aim is to analyze their properties, expressivity and to understand whether they capture a broader range of behaviors than decaying centralities. As it often happens, studying broader and more abstract classes of centralities allows us to understand better the properties of narrower classes. Our work focuses especially on expressivity, following the path opened, for instance, by~\cite{ScBRCSN} in that it attempts to find general results for large sets of centralities.

In particular, the main goal of this work is to answer the following questions:
\begin{itemize}
    \item what can we say about general properties of linear centralities as a function of the vector of coefficients? from an axiomatic perspective: what numerical properties of the vector are related to common axioms, e.g., those discussed in~\cite{BoVAC}?
    \item given a pair of linear geometric centralities, is there always a graph on which they rank nodes differently? in other words: how different are linear centralities from one another?
    \item given a graph, what is the maximum number of rankings of its nodes that linear centralities can induce? how can we characterize precisely those rankings?
\end{itemize} 
The latter problem is particularly interesting, as it has a number of possible applications, especially in an adversarial setting: for instance, if we want to design a graph so that all linear centralities will rank some nodes in a certain way (say, $x$ is always ranked better than $y$, no matter what linear centrality is used), we can employ the results of this paper to determine whether it is possible to do so.

\section{Notations and Basic Definitions}
\label{sec:notations}

For every natural number $n$, we let $[n]=\{0,1,\dots,n-1\}$.
\paragraph{Graphs.} A \emph{graph} $G=(V_G,E_G)$ is given~\cite{BolMGT} by a finite set of nodes $V_G$ and\footnote{For this and similar notations, the subscript $G$ is dropped whenever it is clear from the context.} a set of arcs $E_G\subseteq V_G \times V_G$. 
Without loss of generality, we always assume that $V=[n]$ where $n$ is the number of nodes.

We write $x\to y$ to mean that $x,y\in V$ and $(x,y)\in E$. A graph is \emph{undirected} iff $x\to y$ implies $y\to x$. When drawing undirected graphs, pairs of opposite arcs are represented as a single undirected edge $x\adj y$.

A \emph{path} of length $k$ from $x\in V$ to $y\in V$ is a sequence $(x_0,x_1,\dots,x_k)$ of nodes such that $x=x_0$, $y=x_k$ and $x_i \to x_{i+1}$ for all $i\in [k]$. The \emph{(shortest path) distance} from $x$ to $y$ in $G$, denoted by $d_G(x,y)$, is the length of a shortest path from $x$ to $y$, or $\infty$ if no path from $x$ to $y$ exists.

Note that in this paper, unless otherwise specified, we will consider \emph{directed} graphs, where $d(x,y)$ and $d(y,x)$ can be different (it can even be the case that one is finite and the other is not!). 

\paragraph{Graph homomorphisms.} A \emph{graph homomorphism} $\varphi: G \to H$ is a function $\varphi: V_G \to V_H$ such that $x \to_G y$ iff $\varphi(x) \to_H \varphi(y)$, for all $x,y \in V_G$. A bijective homomorphism $\varphi: G \to H$ is called an \emph{isomorphism}; it is an \emph{automorphism} if $G=H$. A graph with no non-trivial automorphism (i.e., its only automorphism is the identity) is called \emph{rigid}.

\paragraph{Matrices and vectors.} In this paper, we shall use column vectors; all vectors and matrices are conveniently indexed starting from 0. So, if $\vc x \in \RR^n$ and $A\in \RR^{m \times n}$ then

{
    \begin{center}
    $\vc x = 
    \setlength{\arraycolsep}{2pt}
    \renewcommand\arraystretch{1}
    \begin{bmatrix}
        x_0\\
        x_1\\
        \vdots\\ 
        x_{n-1}
    \end{bmatrix}$
    \qquad\text{and}\qquad
    $A =
    \setlength{\arraycolsep}{2pt}
    \renewcommand\arraystretch{1}
    \begin{bmatrix}
        a_{0,0} & a_{0,1} & \dots & a_{0,n-1}\\
        a_{1,0} & a_{1,1} & \dots & a_{1,n-1}\\
        \vdots & \vdots & \ddots & \vdots \\
        a_{m-1,0} & a_{m-1,1} & \dots & a_{m-1,n-1}
    \end{bmatrix}$
    \end{center}
}

\begin{comment}
\[
    \vc x = \left(
    \begin{array}{c}
    x_0\\
    x_1\\
    \dots\\ 
    x_{n-1}
    \end{array}
    \right)
\qquad\text{and}\qquad
    A=\left(
    \begin{array}{llll}
    a_{0,0} & a_{0,1} & \dots & a_{0,n-1}\\
    a_{1,0} & a_{1,1} & \dots & a_{1,n-1}\\
    \dots & \dots \\
    a_{m-11,0} & a_{m-1,1} & \dots & a_{m-1,n-1}
    \end{array}
    \right).
\]
\end{comment}

\paragraph{Infinite vectors and matrix product.} If $\vc a \in \RR^\NN$, we view $\vc a$ as an infinite vector of real numbers. 
If $A\in \RR^{m \times n}$, we write $A\cdot \vc a$ as a shortcut for
\[
    A\cdot \vc a[:\!n]
\]
where $\vc a[:\!n]$ stands for the vector of the first $n$ entries of $\vc a$.

\paragraph{Permutations.} For every $n$, we let $S_n$ be the symmetric group of all permutations (i.e., bijections) $\pi: [n] \to [n]$. The permutation matrix $R_\pi$ is the binary $n \times n$ matrix such that
\[
    \left(R_\pi\right)_{ij}=1 \text{ iff } \pi(i)=j.
\]

\paragraph{Signum.} We define $\sgn: \RR \to \RR$ as the function defined by
\[
    \sgn(x) = \begin{cases}
        -1 & \text{if $x<0$}\\
        0  & \text{if $x=0$}\\
        1  & \text{if $x>0$}.
    \end{cases}
\]

\section{Centralities, Permutations and Agreements}
\label{sec:centr_perm_agree}
A \emph{(graph) centrality} $\ce f$~\cite{BENAMF} is a function associating with each graph $G$ a map $\ce f_G: V_G \to \RR$ such that for any two graphs $G,H$, if $\varphi: G \to H$ is an isomorphism then 
\[
    \ce f_G(x)=\ce f_H(\varphi(x))
\]
for all $x \in V_G$.

While usually centralities are restricted to having non-negative values only, here we are going to relax this requirement and allow for negative centrality values.
As usual, the value $\ce f_G(x)$ is interpreted as a score indicating how central (i.e., important) node $x$ is in $G$: larger values imply greater importance.

We say that $\ce f$ has \emph{no ties on $G$} iff $\ce f_G(x) \neq \ce f_G(y)$ whenever $x,y \in V_G$ and $x\neq y$. Note that a centrality always has ties on non-rigid graphs: if $\phi: G \to G$ is an automorphism then by definition $\ce f(\phi(x))=\ce f(x)$ for all $x$.% such that $\phi(x)\neq x$.

Here are some notable examples of centralities\footnote{Some of these definitions were originally given for undirected graphs only. Here we are providing the adaptations for the general (i.e., directed) case, with the usual proviso that incoming paths are more interesting than outgoing paths.} (see, for instance,~\cite{BoVAC} for further examples and taxonomy):
\begin{itemize}
    \item in-degree: $\cec{in}(x)=|\{y \mid y \to x\}|$;
    \item closeness~\cite{BavMMGS}: $\cec{cl}(x)=1/\sum_{y \in V, d(y,x)<\infty} d(y,x)$;
    \item Lin~\cite{LinFSR}: $\cec{Lin}(x)=|\{y \in V \mid d(y,x)<\infty\}|^2/\sum_{y \in V, d(y,x)<\infty} d(y,x)$;
    \item harmonic~\cite{BoVAC}: $\cec{harm}(x)=\sum_{y \in V} 1/d(y,x)$, with the proviso that $1/\infty=0$;
    \item Katz~\cite{KatNSIDSA}: $\cec{Katz}(x)=\sum_{\pi} \beta^{|\pi|}$, where the sum ranges over all paths $\pi$ ending in $x$, and $\beta$ is a parameter smaller than the reciprocal of the spectral radius of $G$;
    \item betweenness~\cite{AntRG}: $\cec{bet}(x)=\sum_{y,z,\sigma_{yz}\neq0} \sigma_{yz}(x)/\sigma_{yz}$, where the sum ranges over all pairs of nodes $y\neq z$, $\sigma_{yz}$ is the number of shortest paths from $y$ to $z$, and $\sigma_{yz}(x)$ is the number of such paths passing through $x$.
\end{itemize}

In most cases, more than the actual value a centrality assigns to each node, we are interested in how nodes are ranked with respect to a given centrality. This concept is captured formally by the following definition:

\begin{definition}[Respecting a permutation]
Given a graph $G$ of $n$ nodes, we say that $\ce f$ \emph{respects} the permutation $\pi \in S_n$ on $G$ iff $\ce f_G(\pi(i))\ge\ce f_G(\pi(i+1))$ for all $i \in [n-1]$ (that is, $\pi$ orders the nodes of $G$ in non-increasing order of their centrality values with respect to $\ce f$).
\end{definition}

Note that if $\ce f$ has no ties on $G$ then there is a unique permutation that $\ce f$ respects on $G$, called the \emph{permutation induced by $\ce f$}.

\begin{definition}[Agreement between centralities]
Two centralities $\ce f$ and $\ce g$ \emph{agree} on a graph $G$ iff they respect exactly the same permutations on $G$. (In particular, if they both have no ties on $G$, then they must induce the same permutation).
Two centralities that agree on all graphs are called \emph{equivalent}.
\end{definition}

For example, the following two centralities:
\begin{eqnarray*}
    \cec{ccl}(x)&=& \frac{|V|}{\sum_{y \in V, d(y,x)<\infty} d(y,x)}\qquad\text{``classical'' closeness}\\
    \cec{np}(x)&=&-\sum_{y \in V, d(y,x)<\infty} d(y,x)\qquad\text{negative peripherality}
\end{eqnarray*}
are both equivalent to $\cec{cl}$ (closeness). Classical closeness differs from closeness only by a positive multiplicative constant $|V|$, that has no impact on the ranking of nodes. For negative peripherality, instead of taking the reciprocal of the sum of distances to $x$, we take the \emph{opposite} of the sum of distances: although the actual values are no doubt different, $\cec{np}$ induces the same orders on the nodes as the one induced by $\cec{cl}$. 

\section{Geometric and Linear Centralities}

In the wide world of centralities, many depend only on the number of nodes at a(ny) given distance. As explained in Section~\ref{sec:intro}, in this paper we focus on this class.

\begin{definition}[Distance-count function]
Given a graph $G$ with $n$ nodes, and a node $i\in V$, let the \emph{distance-count function of $G$ at $i$} be the function $c_{G,i}:\NN\to\NN$ defined by
\[
    c_{G,i}(k) = |\{j\in V: d(j,i)=k\}|,
\]
that is, the number of nodes at distance $k$ to $i$.
\end{definition}
In~\cite{BH90}, the first $n$ values of the distance-count function at a node $i$ are called \emph{distance-degree sequence} of $i$:
in fact, we do not need to take arbitrary distances into account, because $c_{G,i}(k)=0$ for all $k\geq n$, because two nodes cannot have \emph{finite} distance larger than $n-1$. Moreover, note also that $c_{G,i}(0)=1$ always, because there is exactly one node at distance $0$ from $i$, that is, $i$ itself. Differently from~\cite{KiTTCFDG}, we shall not consider the distance-count function at infinity, i.e., we do not consider the number of nodes that cannot reach $i$.

A centrality is geometric if it is essentially dependent only on distance counts:

\begin{definition}[Geometric centralities]
A centrality $\ce f$ is \emph{strictly geometric} iff for all graphs $G$ and $G'$ and all nodes $i \in V_G$ and $i' \in V_{G'}$ we have
\[
    c_{G,i}=c_{G',i'}   \text{ implies } \ce f_G(i)=\ce f_{G'}(i').
\]
A centrality $\ce f$ is \emph{geometric} iff it is equivalent to one that is strictly geometric.
\end{definition}

For instance, in-degree, closeness, Lin, negative peripherality and harmonic are all strictly geometric: their value on a node $x$ only depends on how many nodes are at each distance from $x$. Classical closeness ($\cec{ccl}$) 
is not strictly geometric: to compute its value on $x$ you need to 
know also the number of nodes in the graph, which is not deducible only from the distance-count function unless the graph is strongly connected. Nonetheless, $\cec{ccl}$ is equivalent to closeness ($\cec{cl}$), so $\cec{ccl}$ is geometric.

Conversely, Katz centrality is not geometric (because it does not depend on shortest paths, but on \emph{all} paths), neither is betweenness (because it does not depend on the length of shortest paths, but on their number).

\smallskip
Geometric centralities depend only on distance counts, but they may do so in many different ways. We are now going to focus further our attention on the case where the dependence is linear. We begin by defining:

\begin{definition}[Distance-count matrix]
The \emph{distance-count matrix} of a graph $G$ of $n$ nodes is the matrix $\,C=C_G \in \RR^{n \times n}$ such that the $i$-th row of $C$ is $\vc c_{G,i}[:\!n]$. Said otherwise, $c_{i,k}$ is the number of nodes at distance $k$ to $i$, for all $i,k \in [n]$.
\end{definition}
Distance-degree sequences represented as matrices also appeared in~\cite{RGRNM} with the name of \emph{neighbor matrices}.

In Figure~\ref{fig:mdc_example}, we show an example of distance-count matrix: node $0$ (first row) has 1 node at distance $0$ (itself), and $n-1$ nodes at distance 1; node $1$ (second row) has 1 node at distance $0$, 1 node at distance 1 (node 0), and $n-2$ nodes at distance 2; and so on.

{
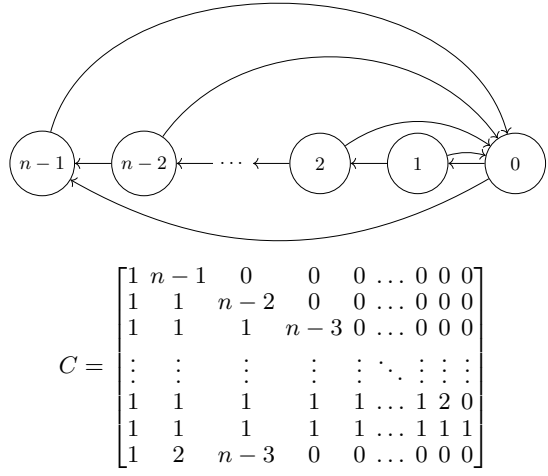
\begin{figure}
    \centering

    \scalebox{0.8}{
    \begin{tikzpicture}[main/.style = {draw, circle, minimum size=10mm}, node distance=0.6cm]
        \node[main] (x0) {$0$};
        \node[main] (x1) [left=of x0] {$1$};
        \node[main] (x2) [left=of x1] {$2$};
        \node (dots) [left=of x2] {$\dots$};
        \node[main] (x3) [left=of dots] {$n-2$};
        \node[main] (x4) [left=of x3] {$n-1$};

        \path[->]
        (x0) edge (x1)
        (x1) edge (x2)
        (x2) edge (dots)
        (dots) edge (x3)
        (x3) edge (x4)
        
        (x1) edge[bend angle=15, bend left] (x0)
        (x2) edge[bend angle=35, bend left] (x0)
        (x3) edge[bend angle=55, bend left] (x0)
        (x4) edge[bend angle=75, bend left] (x0)
        (x0) edge[bend angle=30, bend left] (x4);
    \end{tikzpicture}}
    $C = 
    \setlength{\arraycolsep}{2pt}
    \renewcommand\arraystretch{0.9}
    \begin{bmatrix}
        1 & n-1 & 0 & 0 & 0 & \dots & 0 & 0 & 0 \\
        1 & 1 & n-2 & 0 & 0 & \dots & 0 & 0 & 0 \\
        1 & 1 & 1 & n-3 & 0 & \dots & 0 & 0 & 0 \\
        \vdots & \vdots & \vdots & \vdots & \vdots & \ddots & \vdots & \vdots & \vdots \\
        1 & 1 & 1 & 1 & 1 & \dots & 1 & 2 & 0 \\
        1 & 1 & 1 & 1 & 1 & \dots & 1 & 1 & 1 \\
        1 & 2 & n-3 & 0 & 0 & \dots & 0 & 0 & 0 \\
    \end{bmatrix}$

    \caption{\label{fig:mdc_example}A graph with $n$ nodes, and its distance-count matrix $C$.}
\end{figure}
}

\begin{definition}[Linear centrality]
Given $\vc a \in \RR^\NN$ (the \emph{coefficient vector}), the centrality $\cel{a}$ is defined by
\[
    \cel{a}(i)=\left(C_G \cdot \vc a\right)_i.
\]
A centrality is \emph{strictly linear (geometric)} if it is $\cel{a}$ for some $\vc a \in \RR^\NN$.
A centrality is \emph{linear (geometric)} if it is equivalent to a strictly linear (geometric) centrality.
\end{definition}
For instance,
\begin{itemize}
    \item in-degree is strictly linear, with coefficients $(0,1,0,0,\dots)^T$: its value on $x$ is exactly the number of in-neighbors of $x$ (i.e., nodes at distance $1$ to $x$);
    \item negative peripherality is strictly linear, with coefficients $(0,-1,-2,-3,\dots)^T$;
    \item harmonic is strictly linear, with coefficients $(0,1,1/2,1/3,1/4,\dots)^T$;
    \item closeness and classic closeness are not strictly linear, but they are all linear, since they are all equivalent to negative peripherality.
\end{itemize}
Note that the value of $a_0$ is a constant additive factor that does not affect the ranking of nodes, so we can assume without loss of generality that $a_0=0$, unless otherwise stated.

A support for the computation of general linear centralities has been made available from within the WebGraph framework~\cite{BoVWFI}, a widely used library for the analysis of large graphs. The framework provides \texttt{LinearGeometricCentrality}, a class that can be used to compute any linear centrality, given the coefficient vector $\vc a$.

\subsection{Lin is not linear}
Note that a geometric centrality is not necessarily also linear.
For example, Lin centrality is equivalent to negative peripherality only when the underlying graph is strongly connected, as demonstrated by the following counterexample.

Consider the disconnected undirected graph $G$ in Figure~\ref{fig:lin_counterexample};
if we look at the Lin centrality (refer to Section~\ref{sec:centr_perm_agree} for the formula) of nodes $u$, $v$, $x$, and $y$, we have that:
\[
    \cec{Lin}(u) = \frac{8}{3} < \frac{9}{3} = \cec{Lin}(v) \quad\text{ and }\quad \cec{Lin}(x) = \frac{50}{9} > \frac{49}{9} =  \cec{Lin}(y).
\]
For Lin to be linear, we need to provide a strictly linear centrality to which it is equivalent, i.e., we need a coefficient vector $\vc a$ such that $\cel{a}(u) < \cel{a}(v)$ and $\cel{a}(x) > \cel{a}(y)$.
This means that we need a solution for the following system of inequalities:
\begin{align*}
\begin{cases}
    a_0 + a_1 + a_2 + a_3 < a_0 + a_1 + a_2 \\
    a_0 + 3(a_1 + a_2 + a_3) > a_0 + 3(a_1 + a_2),
\end{cases}
\end{align*}
which obviously has no solution, as it requires both $a_3<0$ and $a_3>0$.

Therefore, there exists no linear centrality equivalent to Lin's centrality on the graph in Figure~\ref{fig:lin_counterexample}.

\begin{figure}
    \centering
    \begin{tikzpicture}[main/.style = {draw, circle, minimum size=3.5mm}, node distance=0.5cm]
        \node[main] (x) {$u$};
        \node[main] (x1) [above=of x] {};
        \node[main] (x2) [above=of x1] {};
        \node[main] (x3) [above=of x2] {};

        \draw (x) -- (x1);
        \draw (x1) -- (x2);
        \draw (x2) -- (x3);

        \node (sp) [right=of x] {};
        \node[main] (y) [right=of sp] {$v$};
        \node[main] (y1) [above=of y] {};
        \node[main] (y2) [above=of y1] {};

        \draw (y) -- (y1);
        \draw (y1) -- (y2);

        \node (sp1) [right=of y] {};
        \node (sp2) [right=of sp1] {};
        \node[main] (u) [right=of sp2] {$x$};
        \node[main] (u12) [above=of u] {};
        \node[main] (u11) [left=of u12] {};
        \node[main] (u13) [right=of u12] {};
        \node[main] (u22) [above=of u12] {};
        \node[main] (u21) [left=of u22] {};
        \node[main] (u23) [right=of u22] {};
        \node[main] (u32) [above=of u22] {};
        \node[main] (u31) [left=of u32] {};
        \node[main] (u33) [right=of u32] {};

        \draw (u) -- (u11);
        \draw (u) -- (u12);
        \draw (u) -- (u13);
        \draw (u11) -- (u21);
        \draw (u12) -- (u22);
        \draw (u13) -- (u23);
        \draw (u21) -- (u31);
        \draw (u22) -- (u32);
        \draw (u23) -- (u33);

        \node (sp3) [right=of u] {};
        \node (sp4) [right=of sp3] {};
        \node (sp5) [right=of sp4] {};
        \node[main] (v) [right=of sp5] {$y$};
        \node[main] (v12) [above=of v] {};
        \node[main] (v11) [left=of v12] {};
        \node[main] (v13) [right=of v12] {};
        \node[main] (v22) [above=of v12] {};
        \node[main] (v21) [left=of v22] {};
        \node[main] (v23) [right=of v22] {};

        \draw (v) -- (v11);
        \draw (v) -- (v12);
        \draw (v) -- (v13);
        \draw (v11) -- (v21);
        \draw (v12) -- (v22);
        \draw (v13) -- (v23);

    \end{tikzpicture}
    \caption{\label{fig:lin_counterexample}A disconnected undirected graph showing that Lin centrality is not linear.}
\end{figure}
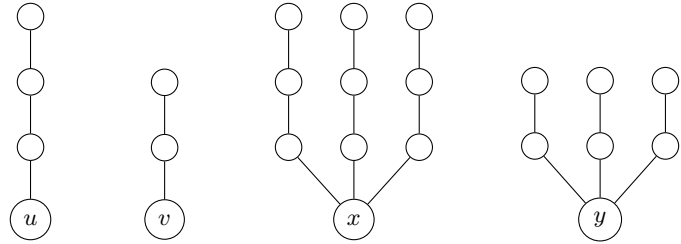

\section{Centralities, Ties and Rigidity}

By its very definition, every centrality (including, of course, all geometric centralities) has ties on all non-rigid graphs. We can say more:

\begin{theorem}
    \label{thm:rigid_noties}
    All centrality measures have ties on graphs that are not rigid. On the other hand, there exists a centrality measure $\ce f$ such that, for all rigid graphs $G$, $\ce f$  has no ties on $G$.
\end{theorem}
\begin{proof}
    We exploit the lexicographic canonization process described in~\cite{babai1983canonical}: let $G$ be a rigid graph with $n$ nodes, and $A_\pi$ be the adjacency matrix of $G$ when its nodes are permuted according to $\pi \in S_n$. Suppose that $A_\pi=A_\rho$ for some choice of permutations $\pi, \rho \in S_n$; then for all $x,y \in V_G$:
    \[
        A_{\pi(x)\pi(y)} = A_{\rho(x)\rho(y)},
    \]
    that is, $\pi(x) \to_G \pi(y)$ iff $\rho(x) \to_G \rho(y)$. But then $\rho^{-1}\circ\pi$ would be an automorphism of $G$, hence $\pi=\rho$. So, the matrices $A_\pi$ are all distinct; hence, there is a unique $\pi$ such that $A_\pi$ is minimal in lexicographic order (we can order matrices lexicographically by considering the string obtained concatenating all of their rows). Now you can define $\ce f_G$ to be such a permutation $\pi: V_G \to [n]$: since $\pi$ is a permutation, by definition it is injective.
\qed\end{proof}

Geometric centralities, though, are more limited and may give rise to unsolvable ties even on rigid graphs. 
For instance, the graph in Figure~\ref{fig:rig_nongrigid} has no non-trivial automorphisms, so according to Theorem~\ref{thm:rigid_noties} there is a centrality that produces no ties on it; nonetheless, all geometric centralities give the same value to nodes $1$ and $2$, because the corresponding rows in its distance-count matrix $C_G$ are the same. Let us provide a formal definition of this limitation:

\begin{definition}[Geometrically rigid]
A graph $G$ is \emph{geometrically rigid} if the rows of its distance-count matrix $C_G$ are all distinct.
\end{definition}

In~\cite{BH90}, geometrically rigid graphs are also called \emph{distance-degree injective}, in opposition to \emph{distance-degree regular} graphs, i.e., those with all nodes having the same distance-degree sequence.

{
\begin{figure}
    \centering
    \begin{tabular}{cccc}
    \scalebox{0.7}{
    \begin{tikzpicture}[main/.style = {draw, circle}, node distance=1cm,baseline=(current bounding box.center)]
        \node[main] (0) at (0, 3) {$0$};
        \node[main] (1) at (0, 0) {$1$};
        \node[main] (2) at (3, 0) {$2$};

        \draw[->=latex] (0) to [bend left] (1);
        \draw[->=latex] (1) to [bend left] (0);
        \draw[->=latex] (1) to (2);
        \draw[->=latex] (2) to (0);
    \end{tikzpicture}}
    & 
    $A_G = 
    \setlength{\arraycolsep}{6pt}
    \renewcommand\arraystretch{1}
    \begin{bmatrix}
        0 & 1 & 0 \\
        1 & 0 & 1 \\
        1 & 0 & 0 
    \end{bmatrix}$
    & $\quad$ &
    $C_G = 
    \setlength{\arraycolsep}{6pt}
    \renewcommand\arraystretch{1}
    \begin{bmatrix}
    1 & 2 & 0 \\
    1 & 1 & 1 \\
    1 & 1 & 1 
    \end{bmatrix}$
    \end{tabular}
    \caption{\label{fig:rig_nongrigid}A rigid graph that is not geometrically rigid, its adjacency matrix $A_G$ and its distance-count matrix $C_G$.}
\end{figure}
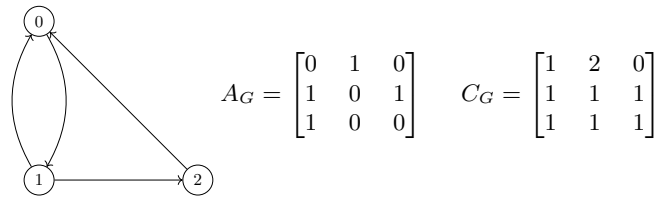
}

Clearly, every geometrically rigid graph is also rigid, but not the other way round, as the graph in Figure~\ref{fig:rig_nongrigid} shows. Nonetheless, we can relate the two notions as follows:

\begin{theorem}
    \label{thm:grigid_noties}
    All geometric centrality measures have ties on graphs that are not geometrically rigid. On the other hand, there exists a geometrical centrality measure $\ce f$ such that, for all geometrically rigid graphs $G$, $\ce f$  has no ties on $G$.
\end{theorem}
\begin{proof}
    Define $\ce f_G(x)$ to be the rank of $x$ in the lexicographic order of the rows of $C_G$. Since by definition the rows of $C_G$ are all distinct, the lexicographic order has no ties. On the other hand, if two rows of $C_G$ are the same, the corresponding nodes will have the same value for all geometric
    centralities (because, by definition, geometric centralities are a function of the distance-count function of the node).
\qed\end{proof}

A similar statement can be proven also for \emph{linear} centralities; in particular, we can show that there is always a linear centrality with no ties on geometrically rigid graphs:

\begin{theorem}
    \label{thm:grigid_noties_linear}
    There exists a choice of coefficients $\vc a \in \RR^\NN$ such that, for all graphs $G$, if $G$ is geometrically rigid then $\cel{a}$ has no ties on $G$.
\end{theorem}
\begin{proof}
    Define $\vc a=(\log p_0,\log p_1,\log p_2,\dots)^T$, where $p_i$ is the $i$-th prime number.
    Then,
    \[
    \cel a(i) = \sum_{j}c_{i,j} \log p_j = \log\Big(\prod_{j} p_j^{c_{i,j}}\Big).
    \]
    Since $\log(x)=\log(y)$ iff $x=y$, $\cel a$ produces a tie iff there are two distinct nodes $i$ and $i'$ such that the products on the right-hand side are equal; clearly, these products are prime factorizations, and this would imply that $c_{i,j}=c_{i',j}$ for all $j$, contradicting the initial hypothesis that $G$ is geometrically rigid.
\qed\end{proof}

\section{Axioms for Linear Centralities}\label{sec:axioms}

While many special instances of linear centralities were studied before, adopting our broader view may be helpful to get a deeper understanding of the most general conditions underlying their properties. As an example in this direction, following a quite common approach, we will consider the axioms described in~\cite{BoVAC} and determine under which general circumstances a linear centrality satisfies them. 

\subsection{Size and density axioms} 

The first two axioms that we are going to study are related to two specific (families of) graphs.

\begin{figure}[htbp]
    \centering
    \begin{tikzpicture}[main/.style = {draw, circle, minimum size=4mm}, scale=1.1]
        \foreach \n in {1,...,8} {
            \ifthenelse{\n = 7}
            {\node[main] (v\n) at (45+\n*45:1.3) {\footnotesize $\,x\,$}}
            {\node[main] (v\n) at (45+\n*45:1.3) {\phantom{\footnotesize $\,x\,$}}};}
        \foreach \i in {1,...,8} {
            \foreach \j in {\i,...,8} {
                \draw (v\i) -- (v\j);}}
        \begin{scope}[shift={(4,0)}]
        \foreach \n in {1,...,8} {
            \ifthenelse{\n = 3}
            {\node[main] (u\n) at (45+\n*45:1.3) {\footnotesize $y$}}
            {\node[main] (u\n) at (45+\n*45:1.3) {\phantom{\footnotesize $y$}}};}
        \tikzset{edge/.style = ->}
        \draw[edge] (u1) -- (u2);
        \draw[edge] (u2) -- (u3);
        \draw[edge] (u3) -- (u4);
        \draw[edge] (u4) -- (u5);
        \draw[edge] (u5) -- (u6);
        \draw[edge] (u6) -- (u7);
        \draw[edge] (u7) -- (u8);
        \draw[edge] (u8) -- (u1);
        \draw[dashed] (u3) -- (v7);
        \end{scope}
        \end{tikzpicture}
    \caption{\label{fig:size_density} A $k$-clique (left) and a directed $p$-cycle (right), possibly connected through a bridge $x\adj y$. We let $S$ be the disconnected graph without the bridge and $S_{xy}$ the strongly connected graph with the bridge.}
\end{figure}
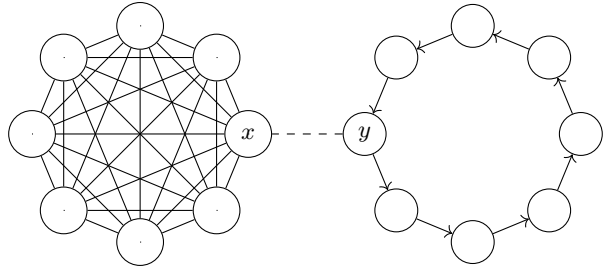

For given positive integers $k,p$, consider the graphs $S$ and $S_{xy}$ defined in Figure~\ref{fig:size_density}. Recall that~\cite{BoVAC}:

\begin{definition}[Density axiom]
    We say that a centrality $\ce f$ satisfies the \emph{density axiom} iff for every fixed $k \geq 3$ and $p=k$, $\ce f_{S_{xy}}(x) > \ce f_{S_{xy}}(y)$.
\end{definition}

\begin{definition}[Size axiom]
    We say that a centrality $\ce f$ satisfies the \emph{size axiom} iff for every fixed $p$, there exists $k_0$ such that for all $k\geq k_0$, $\ce f_S(x) > \ce f_S(y)$, and for every fixed $k$, there exists $p_0$ such that for all $p\geq p_0$, $\ce f_S(x) < \ce f_S(y)$.
\end{definition}

Given a coefficient vector $\vc a$, the geometric linear centralities of $x$ and $y$ in the two graphs $S$ and $S_{xy}$ are given by:
\begin{alignat*}{2}
    \celg{a}{S_{xy}}(x) &= a_0 + k \cdot a_1 + \sum_{i=2}^p a_i && \celg{a}{S}(x) = a_0 + (k-1) \cdot a_1 \\
    \celg{a}{S_{xy}}(y) &= a_0 + 2 \cdot a_1 + (k-1) \cdot a_2 + \sum_{i=2}^{p-1} a_i \qquad\quad && \celg{a}{S}(y) = \sum_{i=0}^{p-1} a_i.
\end{alignat*}

For the density axiom, we have that:
\begin{proposition}\label{prop:density}
$\celg{a}{}$ satisfies the density axiom iff:
\[
        a_k > k(a_2 - a_1) + (2a_1 - a_2)
\]
for all $k \geq 3$.
\end{proposition}
\begin{proof}
We start observing that     
\begin{align*}
        \celg{a}{S_{xy}}(x) - \celg{a}{S_{xy}}(y) & = (k-2) a_1 + (1-k) a_2 + a_p \\
        & = k (a_1 - a_2) + (a_2 - 2a_1) + a_p. 
\end{align*}
So, the condition
\[
        \celg{a}{S_{xy}}(x) > \celg{a}{S_{xy}}(y)
\]
for $p=k$ is equivalent to 
\[
        a_k > k(a_2 - a_1) + (2a_1 - a_2).
\]\qed
\end{proof}

We now turn to the size axiom:

\begin{proposition}\label{prop:size}
    $\cel{a}$ satisfies the size axiom iff $a_1 > 0$ and $\sum_{i} a_i$ diverges.
\end{proposition}
\begin{proof}
    The size axiom requires that
    \begin{equation}
        \label{eqn:sizefixp}
        \celg{a}{S}(x) > \celg{a}{S}(y)
    \end{equation}
    holds ultimately (in $k$) for every fixed $p$, whereas 
    \begin{equation}
        \label{eqn:sizefixk}
        \celg{a}{S}(x) < \celg{a}{S}(y)
    \end{equation}
    should hold ultimately (in $p$) for every fixed $k$.
    This time we have:
    \[
        \celg{a}{S}(x) - \celg{a}{S}(y) = a_0 + (k-1) a_1 - \sum_{i=0}^{p-1} a_i.
    \]
    For~\eqref{eqn:sizefixp} to hold, we must have that for every fixed $p$
    \begin{equation}
        \label{eqn:sizefixpbis}
        (k-1) a_1 > \sum_{i=1}^{p-1} a_i
    \end{equation}
    holds for sufficiently large $k$. If $a_1<0$, there is an upper bound on the largest $k$ for which~\eqref{eqn:sizefixpbis} can hold:
    \[
        k < \frac{\sum_{i=1}^{p-1} a_i}{a_1} + 1.   
    \]
    So we must have $a_1 \geq 0$.
    For~\eqref{eqn:sizefixk} to hold, we must have that for every fixed $k$
    \begin{equation}
        \label{eqn:sizefixkbis}
        (k-1) a_1 < \sum_{i=1}^{p-1} a_i
    \end{equation}
    holds for sufficiently large $p$.  If $a_1=0$, this condition is incompatible with~\eqref{eqn:sizefixpbis}.

    We conclude that $a_1>0$. On the other hand, when $a_1>0$,~\eqref{eqn:sizefixkbis} implies the statement.\qed
\end{proof}

Power-law and exponential decay centralities~\cite{BDFSRSMCBDDC} are special cases of linear centralities;
the former have coefficients of the form $a_i = 1/i^\gamma$ for some $\gamma>0$ and $i>0$, with $\gamma=1$ corresponding to harmonic centrality. 
The latter, instead, have coefficients of the form $a_i = \delta^i$ for some $\delta\in(0,1)$ and $i>0$;
in both cases, as usual, we have $a_0=0$.
From the properties above we can conclude that:

\begin{corollary}
    Power-law decay centralities always satisfy the density axiom, and they satisfy the size axiom iff $\gamma\leq1$.
\end{corollary}
\begin{proof}
    The only non-trivial part of the proof is showing that the density axiom holds.

    We have that $a_k = k^{-\gamma}$, for some $\gamma>0$. Consider the inequality of Proposition~\ref{prop:density}, that in this case becomes:
    \begin{equation}
        \label{eqn:proofpld}
        k^{-\gamma} > k(2^{-\gamma}-1) + 2 - 2^{-\gamma}.
    \end{equation}
    When $k=3$, the right-hand side is $2^{1-\gamma}-1$, which is negative or null for $\gamma \geq 1$ (and a fortiori for larger values of $k$, because the coefficient of $k$ itself is negative), so~\eqref{eqn:proofpld} certainly holds for all $k \geq 3$. 
    We are left to consider the case $\gamma \in (0,1)$. Let us look at the function
    \[
        f(k) = k^{-\gamma} - k(2^{-\gamma}-1) - 2 + 2^{-\gamma}.
    \]
    for $\gamma \in (0,1)$. 
    Since 
    \[
        f'(k)=-\gamma k^{-(\gamma+1)} - (2^{-\gamma}-1),
    \]
    the only zero of the derivative $f'(k)$ is
    \[
        k_0 = \exp\left(-\frac{\log\left(\frac{1-2^{-\gamma}}{\gamma}\right)}{\gamma+1}\right) = \left( \frac\gamma{1-2^{-\gamma}}\right)^{\frac1{\gamma+1}}.
    \]
    Since $f(1)=f(2)=0$, by Rolle's theorem we have that $k_0 \in (1,2)$.
    
    As a result, $f$ is monotonic for $k\geq 3$; its value at 3 is:
    \[
            f(3)=3^{-\gamma}-3(2^{-\gamma}-1) - 2 + 2^{-\gamma} = 3^{-\gamma} - 2^{1-\gamma} + 1 > 0.
    \]
    On the other hand, 
    \[
        \lim_{k \to \infty} f(k)=\infty.
    \]
 So the function $f(k)$ is increasing, hence positive, for $k \geq 3$.\qed
\end{proof}

\begin{corollary}
    Exponential decay centralities always satisfy the density axiom, and never satisfy the size axiom.
\end{corollary}
\begin{proof}
    Once again, the only non-trivial part of the proof is showing that the density axiom holds.
    The inequality of Proposition~\ref{prop:density} in this case is:
    \[
        \delta^k > k(\delta^2 - \delta) + 2 \delta - \delta^2;
    \]
    thus, we want to show that
    \[
        f(k) = \delta^k + k(\delta - \delta^2) + \delta^2 - 2\delta 
    \]
    is positive for all $\delta\in(0,1)$ and $k\geq 3$.
    The derivative of the function is
    \[
        f'(k) = \delta^k \log\delta + \delta - \delta^2;
    \]
    let us call $k_0$ the value in which $f'(k)$ is null, that is,
    \[
        k_0 = \frac{\log\left(\frac{\delta^2 - \delta}{\log(\delta)}\right)}{\log(\delta)}.
    \]

    Noting that $f(1)=f(2)=0$, we can state by Rolle's theorem that $k_0\in(1,2)$, since it is the only stationary point of the function $f$.
    Finally, we highlight that
    \[
        f(3) = \delta^3 + 3(\delta - \delta^2) + \delta^2 - 2\delta = \delta^3 - 2\delta^2 + \delta =\delta(\delta-1)^2 > 0,
    \]
    and
    \[
        \lim_{k\to\infty} f(k) = \infty.
    \]
    Again, the function is increasing, hence positive, for $k\geq 3$.\qed

\end{proof}

In particular, among power-law and exponential decay centralities, harmonic has the property of satisfying both axioms, and it is the one doing so with the fastest possible decay.

\subsection{Score monotonicity}

We now consider the score monotonicity axiom~\cite{BoVAC}, that is related to what happens to the centrality of node $y$ when an arc $x \to y$ is added to an arbitrary graph.

\begin{definition}[Score monotonicity]
    A centrality measure is \emph{score monotone} if for every graph $G$ and every pair of nodes $x$, $y$ such that $x\not\rightarrow y$, when we add $x\rightarrow y$ to $G$ the centrality of $y$ increases.
\end{definition}

To characterize score monotonicity, it is useful to introduce (as  we did in~\cite{BDFSRSMCBDDC}) the \emph{discrete derivative operator} $\Delta$ on sequences: given a sequence $\vc a$, we define $\Delta \vc a$ as the sequence
\[
    \left(\Delta \vc a\right)_i = a_{i+1} - a_i
\]
for all indices $i>0$, and $(\Delta \vc a)_0=0$.
It is also useful to write $\vc a \geq 0$ ($\vc a \leq 0$, respectively) to mean that $a_i \geq 0$ ($a_i \leq 0$, resp.) for all $i$.

Score monotonicity for linear centrality can then be characterized by the following:

\begin{proposition}
    \label{prop:scoremon}
    Given a coefficient vector $\vc a$, the centrality $\celg{a}{}$ is score monotone iff $\vc a \geq 0$, $\Delta \vc a \leq 0$ and $\left(\Delta \vc a\right)_1 < 0$.
\end{proposition}
\begin{proof}
    The condition can be restated as follows: $a_1 > a_2$ and $a_i \geq a_{i+1} \geq 0$ for all $i$.
    Given a graph $G$ and a pair of nodes $x,y$ such that $x\not\rightarrow y$, let us call $G'$ the graph obtained from $G$ by adding the arc $x\rightarrow y$.
    Let us assume that $a_1 > a_i \geq a_{i+1} \geq 0$ for all $i\geq2$, and call $d'(z)=d_{G'}(z,y)$ and $d(z)=d_{G}(z,y)$ for all $z\in V_G$.
    Adding an arc in a graph cannot increase distances, that is:
    \begin{equation*}
        d'(z) \leq d(z) \text{ for all } z\in V_G,
    \end{equation*}
    which implies that $a_{d'(z)} \geq a_{d(z)}$, with the two inequalities being strict when $z=x$.
    Noting that linear centralities can be expressed also as summation over nodes, we can state that:
    \begin{align*}
        \celg{a}{G'}(y) = \sum_{z\in V_G} a_{d'(z)} 
        > \sum_{z\in V_G} a_{d(z)} = \celg{a}{G}(y),
    \end{align*}
    where, again, the inequality is strict because it is so when $z=x$.

    For the ``only if'' part of the proof, we consider the following cases, and provide for each of them an example on which $\celg{a}{}$ fails to be score monotone.
    
    \noindent\emph{(a) Case $a_1\leq 0$.} Suppose $a_1 \leq 0$ and consider the following graph $G$: 
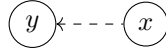
    \[
        \begin{tikzpicture}[main/.style = {draw, circle}, node distance=1cm]
            \node[main] (y) at (0, 0) {$y$};
            \node[main] (x) at (1.5, 0) {$x$};
            \draw[<-,dashed] (y)--(x);
        \end{tikzpicture}
    \]
    We would then have
    \[
        \celg{a}{G'}(y)=a_1 \leq 0 = \celg{a}{G}(y).
    \]

    \noindent\emph{(b) Case $a_i<0$.}
    Suppose $a_1>0$ but $a_i<0$ for some $i>1$. Let $i$ be the smallest such index, and $k$ be defined by
    \[
        k = \left\lceil\frac{a_1+\dots+a_{i-1}+1}{|a_i|}\right\rceil
    \]
    and consider the following graph $G$:
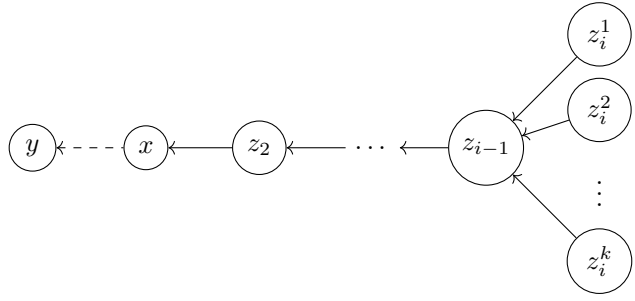
    \[
        \begin{tikzpicture}[main/.style = {draw, circle}, node distance=1cm]
            \node[main] (z0) at (0, 0) {$y$};
            \node[main] (z1) at (1.5, 0) {$x$};
            \node[main] (z2) at (3, 0) {$z_2$};
            \node (z3) at (4.5, 0) {\dots};
            \node[main] (zim) at (6, 0) {$z_{i-1}$};
            \node[main] (zi1) at (7.5, 1.5) {$z_i^1$};
            \node[main] (zi2) at (7.5, 0.5) {$z_i^2$};
            \node (zi3) at (7.5, -0.5) {$\vdots$};
            \node[main] (zik) at (7.5, -1.5) {$z_i^k$};
            \draw[<-,dashed] (z0)--(z1);
            \draw[<-] (z1)--(z2);
            \draw[<-] (z2)--(z3);
            \draw[<-] (z3)--(zim);
            \draw[<-] (zim)--(zi1);
            \draw[<-] (zim)--(zi2);
            \draw[<-] (zim)--(zik);
        \end{tikzpicture}
    \]
    We have that:
    \begin{multline*}
        \celg{a}{G'}(y)=a_1+a_2+\dots+a_{i-1}+k \cdot a_i =\\
        =a_1+a_2+\dots+a_{i-1}-k \cdot |a_i| \leq\\ 
        \leq a_1+a_2+\dots+a_{i-1}-(a_1+\dots+a_{i-1}+1) < 0 = \celg{a}{G}(y).
    \end{multline*}

    \noindent\emph{(c) Case $a_i<a_{i+1}$.}
    To complete our proof, suppose that $a_1>0$ and $a_i \geq 0$ for all $i>1$, but
    that $a_i < a_{i+1}$ for some $i \geq 1$; let $i$ be the smallest such index $i$, and $k$ be defined by:
    \[
        k = \left\lceil\frac{a_1+1}{a_{i+1}-a_i}\right\rceil
    \]
    Consider the following graph $G$.
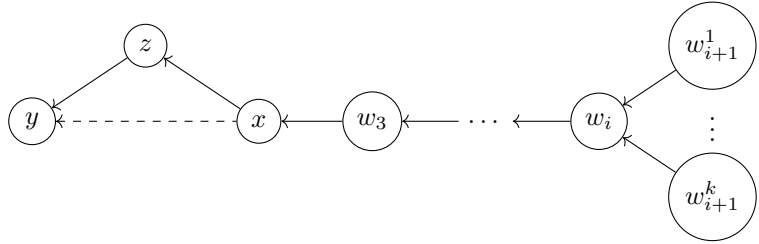
    \[
        \begin{tikzpicture}[main/.style = {draw, circle}, node distance=1cm]
            \node[main] (y) at (0, 0) {$y$};
            \node[main] (z) at (1.5, 1) {$z$};
            \node[main] (x) at (3, 0) {$x$};
            \node[main] (w2) at (4.5, 0) {$w_3$};
            \node (w3) at (6, 0) {\dots};
            \node[main] (wi) at (7.5, 0) {$w_i$};
            \node[main] (wip11) at (9, 1) {$w_{i+1}^1$};
            \node (wip13) at (9, 0) {$\vdots$};
            \node[main] (wip1k) at (9, -1) {$w_{i+1}^k$};
            \draw[<-,dashed] (y)--(x);
            \draw[<-] (y)--(z);
            \draw[<-] (z)--(x);
            \draw[<-] (x)--(w2);
            \draw[<-] (w2)--(w3);
            \draw[<-] (w3)--(wi);
            \draw[<-] (wi)--(wip11);
            \draw[<-] (wi)--(wip1k);
        \end{tikzpicture}
    \]
    We have:
    \begin{eqnarray*}
        \celg{a}{G}(y) &=& a_1+a_2+\dots+k\cdot a_{i+1}\\
        \celg{a}{G'}(y) &=& 2a_1+a_2+\dots+k\cdot a_{i}.
    \end{eqnarray*}
    Hence
    \[
        \celg{a}{G}(y)-\celg{a}{G'}(y) = -a_1+a_i+k(a_{i+1}-a_i) > 0,
    \]
    that is, $\celg{a}{G'}(y)<\celg{a}{G}(y)$.\qed
\end{proof}

\subsection{Rank monotonicity}

The rank monotonicity axiom~\cite{BLVRMCM} considers the same scenario as for score monotonicity, but this time looking at what happens to the relation between the centrality of $y$ and that of \emph{other} nodes.

\begin{definition}[Rank monotonicity]
    A centrality measure is \emph{rank monotone}\footnote{This is called ``strict'' rank monotonicity in~\cite{BLVRMCM}, but in this paper we are not going to treat the weak case.} if for every graph $G$ and every pair of nodes $x$, $y$ such that $x\not\rightarrow y$, when we add $x\rightarrow y$ to $G$ the following happens for all nodes $w$:
    \begin{itemize}
        \item if the score of $w\neq y$ was smaller than or equal to the score of $y$, after adding $x\rightarrow y$ the score of $w$ is smaller than the score of $y$.
    \end{itemize}
\end{definition}

The first property we prove is that score monotonicity is necessary for rank monotonicity:

\begin{lemma}
    \label{lemma:twocopies}
    Given a coefficient vector $\vc a$, if $\celg{a}{}$ is rank monotone, then it is score monotone.
\end{lemma}
\begin{proof}
    Suppose, by contradiction, that $\celg{a}{}$ is rank monotone but not score monotone, and let $G$ be a graph for which score monotonicity fails.
    Let $\overline G$ be an identical copy of $G$ (the nodes of $\overline G$ have the same names as those of $G$, with an $\overline{\cdot}$ to distinguish them).
    For the moment, assume that $a_1 > 0$, and consider the graph:
    \[
        H = G \oplus \overline G 
    \]
    where $\oplus$ denotes disjoint union.
    Of course
    \[
        \celg{a}{H}(y) = \celg{a}{G}(y) = \celg{a}{\overline G}(\overline y)=\celg{a}{H}(\overline y).
    \]
    If we add $x \to y$ to the graph $H$, without any other modification, and let $H'$ be the new graph, we have (because score monotonicity does not hold):
    \[
        \celg{a}{H'}(y) = \celg{a}{G'}(y) \le \celg{a}{G}(y) = \celg{a}{H}(\overline y).
    \]
    So the score of $y$ in $H$ used to be $\geq$ (in fact: equal) to the score of $\overline y$, but it should become strictly higher once the arc $x \to y$ is added, which does not happen if score monotonicity fails.\qed
\end{proof}

Hence, if we want to look for vectors $\vc a$ for which rank monotonicity holds, we can limit ourselves to those for which score monotonicity holds as well.

Like in~\cite{BDFSRSMCBDDC}, we will show that rank monotonicity is related to the second-order discrete derivative $\Delta^2 \vc a$ of the coefficient vector $\vc a$ (obtained from $\vc a$ by applying the operator $\Delta$ twice).
To characterize rank monotonicity, we will need the following lemma:

\begin{lemma}[Lemma 5.1 of~\cite{BDFSRSMCBDDC}]
    \label{lemma:convexity}
    Given a coefficient vector $\vc a$, if $\Delta^2 \vc a \geq 0$ then for all $i,j,k,\ell$ with $i \leq j$, $k \geq 0$, $0\leq \ell \leq j-i$ we have:
    \[
        a_{i+k} - a_{j+k-\ell} \leq a_i - a_j
    \]
    with the inequality being strict if $i<j$, $k>0$ and $\left(\Delta^2 \vc a\right)_i >0$.
\end{lemma}

We now present our conditions for rank monotonicity:
\begin{proposition}
    \label{prop:rankmon}
    Suppose that the coefficient vector $\vc a$ satisfies the following conditions:
    \begin{enumerate}
        \item $\celg{a}{}$  is score monotone, i.e., $\vc a \geq 0$, $\Delta \vc a \leq 0$ and $\left(\Delta \vc a\right)_1 < 0$;
        \item $\Delta^2 \vc a \geq 0$ and $\left(\Delta^2 \vc a\right)_1 > 0$.
    \end{enumerate}
    Then, the centrality $\celg{a}{}$ is rank monotone.
\end{proposition}
\begin{proof}
    Suppose that the two conditions hold, let $G$ be a graph and $G'$ the graph obtained after the addition of $x \to y$ to $G$, as usual. Let us use $d(-,-)$
    and $d'(-,-)$ to denote distances in $G$ and $G'$, respectively. 

    Consider any two nodes $z$ and $w$: of course $d'(z,w) \leq d(z,w)$ (because adding an arc cannot make distances larger); suppose that $d'(z,w)<d(z,w)$. 
    This means that there is a new shortest path from $z$ to $w$ using the new arc $x \to y$. Hence
    \begin{equation}
        \label{eqn:dleq}
        d'(z,w) = d(z,x) + 1 + d(y,w) \geq d'(z,y) + d(y,w).
    \end{equation}
    On the other hand, because of the triangular inequality
    \begin{equation}
        \label{eqn:triangular}
        d(z,w) \leq d(z,y)+d(y,w);
    \end{equation}
    hence, applying~\eqref{eqn:dleq} first and then~\eqref{eqn:triangular}, we have
    \begin{equation}
        \label{eqn:dleq2}
        d'(z,y) \leq  d'(z,w)-d(y,w) \leq d'(z,w)-d(z,w)+d(z,y).
    \end{equation}
    Let
    \begin{eqnarray*}
        i &=& d'(z,y)\\
        j &=& d(z,y)\\
        k &=& d'(z,w) - i = d'(z,w)-d'(z,y)\\
        \ell &=& j + k - d(z,w) = d'(z,w)-d(z,w)+d(z,y)-d'(z,y).
    \end{eqnarray*}
    Let us check that all conditions of Lemma~\ref{lemma:convexity} hold: $i\leq j$ holds because distances cannot become larger; $k\geq0$ holds because of~(\ref{eqn:dleq}); finally $\ell\geq 0$ (because of~\eqref{eqn:dleq2}) and
    \[
        \ell+i-j =i+k-d(z,w)=d'(z,w)-d(z,w) \leq 0;
    \]
    therefore, $\ell\leq j-i$.
    Hence, applying the lemma we have:
    \begin{eqnarray}
        \label{eqn:score_increase}
        %a_{i+k} - a_{j+k-\ell} &\leq& a_{i} - a_{j}\\
        a_{d'(z,w)} - a_{d(z,w)} &\leq& a_{d'(z,y)} - a_{d(z,y)}.        
    \end{eqnarray}
    This relation holds for all $z$ such that $d'(z,w)<d(z,w)$, but it is trivially true also when $d'(z,w)=d(z,w)$, because the sequence $\vc a$ is non-increasing.

    Moreover, when $z=x$ and $w \neq y$, if $d'(x,w)<d(x,w)$ we have $i=d'(x,y)=1$, $j=d(x,y)>1=i$, $k=d'(x,w)-1>0$. Hence, the inequality~\eqref{eqn:score_increase} holds strictly for $z=x$. If otherwise $d'(x,w)=d(x,w)$, then still the inequality holds strictly, because $a_1>a_2 \geq a_j$, since $\left(\Delta \vc a\right)_1<0$.
    
    Recalling that 
    \[
        \celg{a}{G}(w)  = \sum_{z} a_{d(z,w)}  \qquad
        \celg{a}{G'}(w)  = \sum_{z} a_{d'(z,w)},
    \]
    we can state the following for all $w \neq y$:
    \begin{align*}
        \celg{a}{G'}(y)-\celg{a}{G}(y) & = \sum_{z} (a_{d'(z,y)}-a_{d(z,y)}) \\
        & > \sum_{z} (a_{d'(z,w)}-a_{d(z,w)}) = \celg{a}{G'}(w)-\celg{a}{G}(w).
    \end{align*}
    This means that the addition of the arc $x\to y$ increases the score of $y$ more than the score of any other node $w$, and this is enough to conclude that rank monotonicity holds for $\celg{a}{}$.\qed
\end{proof}

Observe that the conditions of Proposition~\ref{prop:rankmon} are sufficient for rank monotonicity, but not necessary (although score monotonicity is). We conjecture that they are also necessary: a clue in this direction is given in~\cite{BDFSRSMCBDDC}, where convexity is proven to be equivalent in the undirected case to a property known as ``basin dominance'', which implies rank monotonicity. But also in that case we were unable to prove that convexity is necessary for rank monotonicity.

\section{Distinguishable Coefficients}
\label{sec:distinguish}

A natural question we may want to ask is the following: how does $\celg{a}{G}$ change as a function of $\vc a$?
Is it possible that using different coefficients we end up defining equivalent linear centralities?
The short answer is no, unless the coefficient vectors differ by a positive multiplicative constant.
Let us provide a specific definition for this case.

\begin{definition}[Proportionality]
    \label{def:proportionality}
    We say that $\vc a, \vc b \in \RR^\NN$ are \emph{proportional} iff there is a $\lambda>0$ such that
    \begin{equation}
    \label{eqn:prop}
        b_k = \lambda \cdot a_k   \text{ for all $k>0$}.
    \end{equation}
\end{definition}

Observe that in the definition the 0-th coefficient has no role. For this reason, and for what we have said before, in the rest of this section we can safely assume that the $0$-th coefficient of all sequences is $0$.
Some observations are in order:
\begin{itemize}
    \item if $\vc a$ and $\vc b$ are proportional then for all $k>0$, either $a_k$ and $b_k$ are both zero, or they are both non-zero (and have the same sign);
    \item no sequence is proportional to $\vc 0=(0,0,0,\dots)$, except $\vc 0$ itself;
    \item if $\vc a$ and $\vc b$ are proportional, and $\vc a \neq \vc 0$, then the constant $\lambda>0$ of Definition~\ref{def:proportionality} is unique;
    %\item if $\vc a$ and $\vc b$ are not proportional, there exist an index $k>1$ and $\lambda>0$ such that $a_i=\lambda\cdot b_i$ for all $0 <i<k$, but $a_k\neq \lambda\cdot b_k$.
    \item if $\vc a$ and $\vc b$ are not proportional, one of the following mutually exclusive statements holds:
    \begin{itemize}
        \item[(a)] there is some $k>0$ such that $a_kb_k=0$ and $a_k\neq b_k$ (for example, $\vc a=(0,1,2,\dots)$ and $\vc b=(0,1,0,\dots)$, with $k=2$);
        \item[(b)] for all $i>0$, $a_i=0 \iff b_i=0$ and the minimum index $h$ such that $a_h\neq 0$ satisfies $b_h=\lambda\cdot a_h$ for some $\lambda<0$ (for example, $\vc a=(0,0,2,3,0,\dots)$ and $\vc b=(0,0,-4,-6,0,\dots)$, with $h=2$ and $\lambda=-2$); 
        \item[(c)] for all $i>0$, $a_i=0 \iff b_i=0$ and there exists an index $k>1$ and $\lambda>0$ such that $b_i=\lambda\cdot a_i$ for all $0 <i<k$, but $b_k\neq \lambda\cdot a_k$, and moreover $a_h\neq 0$ for some $0<h<k$ (for example, $\vc a=(0,1,2,3,0,\dots)$ and $\vc b=(0,2,4,9,0,\dots)$, with $k=3$ and $\lambda=2$).
    \end{itemize}
\end{itemize}

The only non-trivial observation is the last one, which we are now going to prove.
If there is any index for which either sequence is non-zero and the other is zero (case (a) above), there cannot exist a $\lambda>0$ for Definition~\ref{def:proportionality}. Now, let us assume that the zeroes are aligned (i.e., that they appear in the same positions in both sequences): there must be an index at which the sequences are both non-zero (for otherwise both sequences would be $\vc 0$ and proportional). Consider the minimum index $h$ at which they are both non-zero, and let $\lambda=a_h/b_h$. If $\lambda<0$, then case (b) of the last observation holds. If $\lambda>0$, since the sequences are non-proportional there must be a larger index $k$ such that $a_k \neq \lambda b_k$; the smaller such $k>h$ makes the statement of case (c) above true.

It is worth noting that case (c) above has some interesting implications:

\begin{lemma}
\label{lemma:altc}
    Given $\vc a$ and $\vc b$, if case (c) of the previous observation holds, then the values $a_h, a_k, b_h, b_k$ are all non-zero and $b_h a_k \neq b_k a_h$.
    Moreover, the inequality
    \begin{equation}
    \label{eqn:nonprop_cond}
        (a_h-a_k)(b_{h+1}-b_{k-1}) \neq (a_{h+1}-a_{k-1})(b_h-b_k)
    \end{equation}
    holds iff $a_{h+1} \neq a_{k-1}$ and $h+1<k$.
\end{lemma}

\begin{proof}
    We start by proving the first part of the statement.
    By assumption $a_h\neq 0$ and $b_h=\lambda a_h \neq 0$.
    Moreover, either $a_k$ and $b_k$ are both zero, or none of them is;
    if they are, then $b_k= \lambda a_k$, contradicting the hypothesis.
    Finally,
    \[
        b_h a_k = \lambda a_h a_k \neq a_h b_k.
    \]

    For the second part of the statement, assume that $a_{h+1} \neq a_{k-1}$ and $h+1<k$, and, by contradiction, consider the following chain of equalities:
    \begin{align*}
        (a_h-a_k)(b_{h+1}-b_{k-1}) & = (a_{h+1}-a_{k-1})(b_h-b_k)\\
        (a_h-a_k)(\lambda a_{h+1}-\lambda a_{k-1}) & = (a_{h+1}-a_{k-1})(\lambda a_h-b_k) \\
        \lambda a_h a_{h+1} - \lambda a_h a_{k-1} - \lambda a_{h+1} a_k + \lambda a_{k-1} a_k & = \lambda a_h a_{h+1} - a_{h+1} b_k - \lambda a_h a_{k-1} + a_{k-1} b_k \\
        - \lambda a_{h+1} a_k + \lambda a_{k-1} a_k & = - a_{h+1} b_k + a_{k-1} b_k \\
        \lambda a_k (a_{k-1} - a_{h+1}) & = b_k (a_{k-1} - a_{h+1}) \\
        \lambda a_k & = b_k,
    \end{align*}
    contradicting the assumptions.
    For the \q{only if} part, if $h+1=k$, we obtain
    \[
        (a_h-a_k)(b_k-b_h) = (a_k-a_h)(b_h-b_k),
    \]
    contradicting~\eqref{eqn:nonprop_cond}.
    On the other hand, if $h+1<k$ and $a_{h+1} = a_{k-1}$, then we have
    $a_{h+1} - a_{k-1}=0$ but also
    $b_{h+1} - b_{k-1} =\lambda (a_{h+1}-a_{k-1})=0$, again contradicting~\eqref{eqn:nonprop_cond}.\qed
\end{proof}

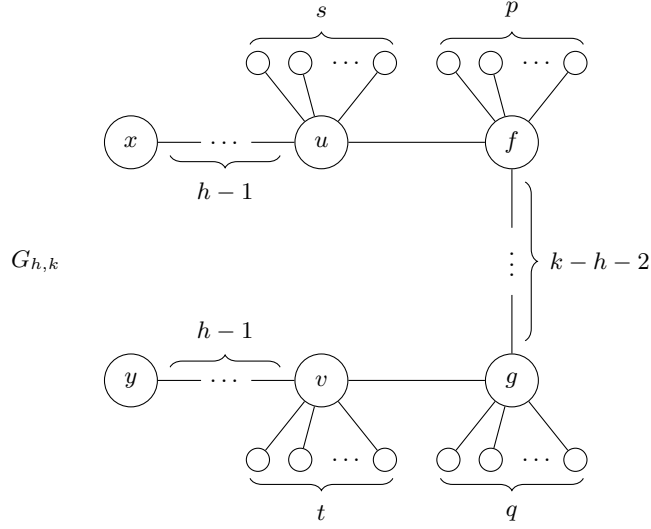
\begin{figure}[htbp]
    \centering
    \begin{tikzpicture}[main/.style = {draw, circle, minimum size=7mm}, node distance=1cm, scale=.7]

        \tikzmath{\d1=2.0; \d2=1.5; \s1=1.0; \d3=-4.5;}

        \node[] (G3) at (-9,\d2+\d3/2) {$G_{h,k}$};
        \node[main] (x1) at (\d1*-3.6,\d2) {$x$};
        \node (xh1) at (\d1*-2.7,\d2) {$\dots$};
        \node[main] (h1) at (\d1*-1.8,\d2) {$u$};
        %\node (hf1) at (\d1*-0.9,\d2) {$\dots$};
        \node[main] (f1) at (\d1*0,\d2) {{\small$f$}};

        \draw (x1) -- (xh1);
        \draw (xh1) -- (h1);
        \draw (h1) -- (f1);

        \node[main, minimum size=3mm] (s1) at (\d1*-1.8-\s1*1.2,\d2*2) {};
        \node[main, minimum size=3mm] (s2) at (\d1*-1.8-\s1*0.4,\d2*2) {};
        \node (s3) at (\d1*-1.8+\s1*0.5,\d2*2) {$\dots$};
        \node[main, minimum size=3mm] (s4) at (\d1*-1.8+\s1*1.2,\d2*2) {};

        \draw (h1) -- (s1);
        \draw (h1) -- (s2);
        \draw (h1) -- (s4);

        \draw [decoration = {raise=-3pt,brace,amplitude=0.5em},decorate,black] (-0.15+\d1*-1.8-\s1*1.2,\d2*2+0.5) -- (+0.15+\d1*-1.8+\s1*1.2,\d2*2+0.5) node[pos=0.5,above=4pt,black] {$s$};

        \node[main, minimum size=3mm] (p1) at (\d1*0-\s1*1.2,\d2*2) {};
        \node[main, minimum size=3mm] (p2) at (\d1*0-\s1*0.4,\d2*2) {};
        \node (p3) at (\d1*0+\s1*0.5,\d2*2) {$\dots$};
        \node[main, minimum size=3mm] (p4) at (\d1*0+\s1*1.2,\d2*2) {};

        \draw (f1) -- (p1);
        \draw (f1) -- (p2);
        \draw (f1) -- (p4);

        \draw [decoration = {raise=-3pt,brace,amplitude=0.5em},decorate,black] (-0.15+\d1*0-\s1*1.2,\d2*2+0.5) -- (+0.15+\d1*0+\s1*1.2,\d2*2+0.5) node[pos=0.5,above=4pt,black] {$p$};
        \draw [decoration = {raise=-3pt,brace,mirror,amplitude=0.5em},decorate,black] (\d1*-3.23,\d2-0.4) -- (\d1*-2.19,\d2-0.4) node[pos=0.5,below=4pt,black] {$h-1$};
        %\draw [decoration = {raise=-3pt,brace,amplitude=0.5em},decorate,black] (\d1*-0.4,\d2-0.4) -- (\d1*-1.4,\d2-0.4) node[pos=0.5,below=4pt,black] {$k-h$};

        % second graph

        \node[main] (x2) at (\d1*-3.6,\d2+\d3) {$y$};
        \node (xh2) at (\d1*-2.7,\d2+\d3) {$\dots$};
        \node[main] (h2) at (\d1*-1.8,\d2+\d3) {$v$};
        %\node (hf2) at (\d1*-0.9,\d2+\d3) {$\dots$};
        \node[main] (f2) at (\d1*0,\d2+\d3) {{\small$g$}};

        \node (dots) at (\d1*0,\d2+\d3/2) {$\myvdots$};
        \draw (f1) -- (dots);
        \draw (dots) -- (f2);
        \draw [decoration = {raise=0-3pt,brace,amplitude=0.5em},decorate,black] (\d1*0+0.4,\d2-0.75) -- (\d1*0+0.4,\d2+\d3+0.75) node[pos=0.5,right=3pt,black] {$k-h-2$};

        \draw (x2) -- (xh2);
        \draw (xh2) -- (h2);
        \draw (h2) -- (f2);

        \node[main, minimum size=3mm] (s12) at (\d1*-1.8-\s1*1.2,\d3) {};
        \node[main, minimum size=3mm] (s22) at (\d1*-1.8-\s1*0.4,\d3) {};
        \node (s32) at (\d1*-1.8+\s1*0.5,\d3) {$\dots$};
        \node[main, minimum size=3mm] (s42) at (\d1*-1.8+\s1*1.2,\d3) {};

        \draw (h2) -- (s12);
        \draw (h2) -- (s22);
        \draw (h2) -- (s42);

        \draw [decoration = {raise=0-3pt,brace,amplitude=0.5em},decorate,black] (+0.15+\d1*-1.8+\s1*1.2,-0.5+\d3) -- (-0.15+\d1*-1.8-\s1*1.2,-0.5+\d3) node[pos=0.5,below=4pt,black] {$t$};

        \node[main, minimum size=3mm] (p12) at (\d1*0-\s1*1.2,\d3) {};
        \node[main, minimum size=3mm] (p22) at (\d1*0-\s1*0.4,\d3) {};
        \node (p32) at (\d1*0+\s1*0.5,\d3) {$\dots$};
        \node[main, minimum size=3mm] (p42) at (\d1*0+\s1*1.2,\d3) {};

        \draw (f2) -- (p12);
        \draw (f2) -- (p22);
        \draw (f2) -- (p42);

        \draw [decoration = {raise=-3pt,brace,amplitude=0.5em},decorate,black] (+0.15+\d1*0+\s1*1.2,-0.5+\d3) -- (-0.15+\d1*0-\s1*1.2,-0.5+\d3) node[pos=0.5,below=4pt,black] {$q$};
        
        \draw [decoration = {raise=-3pt,brace,mirror,amplitude=0.5em},decorate,black] (\d1*-2.19,\d2+0.4+\d3) -- (\d1*-3.23,\d2+0.4+\d3) node[pos=0.5,above=4pt,black] {$h-1$};
        %\draw [decoration = {raise=-3pt,brace,amplitude=0.5em},decorate,black] (\d1*-1.4,\d2+0.4+\d3) -- (\d1*-0.4,\d2+0.4+\d3) node[pos=0.5,above=4pt,black] {$k-h$};
    \end{tikzpicture}
    \caption{The (family of) graphs $G_{h,k}$ used to distinguish non-proportional sequences when inequality~\eqref{eqn:nonprop_cond} holds. The horizontal dotted line between $x$ and $u$ (resp., $y$ and $v$) represents a path of length $h-1$.
    Analogously, the vertical dotted line between $f$ and $g$ represents a path of length $k-h-2$.
    Note that when $h=1$, $u$ gets to coincide with $x$, hence the $s$ nodes in such a case are connected to $x$ (everything holds symmetrically for $y$).}\label{fig:inf_seq_graph}
\end{figure}

Based on these observations, we are ready to prove the following:
\begin{theorem}
    \label{thm:distinguish}
    For $\vc a, \vc b \in \RR^\NN$ the following holds:
    \begin{enumerate}
        \item if $\vc a$ and $\vc b$ are proportional, then $\celg{\vc a}{}$ and $\celg{\vc b}{}$ are equivalent centralities;
        \item if $\vc a$ and $\vc b$ are not proportional, then there exists a graph $G$ on which  $\cel{a}$ and $\cel{b}$ do not agree; in other words, there is a graph $G$ and two nodes $x,y \in V$ such that $\cel{a}(x)\geq \cel{a}(y)$ but $\cel{b}(x)<\cel{b}(y)$. 
        %In such a case, we say that $G$ \emph{distinguishes} $\vc a$ and $\vc b$.
    \end{enumerate}
\end{theorem}
\begin{proof}
    Given a graph $G=(V,E)$, if $\vc a$ and $\vc b$ are proportional, then $\cel{b}(x) = \lambda \cdot \cel{a}(x)$ for some $\lambda>0$ and all $x \in V$, and it is easy to see that in such a case the two centralities are equivalent.

    Let us now assume that $\vc a$ and $\vc b$ are not proportional, and look at the observations before Lemma~\ref{lemma:altc}.
    If we are in case (a), and $k$ is the first index at which either sequence is zero and the other is not, let us consider a graph made of two bidirectional paths of length $k-1$ and $k$, respectively.
    Assuming w.l.o.g.~that $b_k \neq 0$, if $b_k > 0$ the statement is proven by letting $x$ ($y$, respectively) be at the extreme end of the path of length $k-1$ ($k$, respectively).
    %If $b_k < 0$, let $x$ ($y$, respectively) be at the extreme end of the path of length $k$ ($k-1$, respectively).
    If $b_k < 0$, just flip $x$ and $y$, and everything holds thanks to the negative contribution of $b_k$ to $\cel{b}(x)$.
    If we are in case (b), instead, we just need a single bidirectional path of length $h$ with $x$ and $y$ being the two consecutive nodes at one extreme end of the path, with their position depending on the sign of $a_h$ and $b_h$.
    
    We are left to show the statement for two non-proportional $\vc a$ and $\vc b$ both different from $\vc 0$ that satisfy case (c) of the observation above.
    We can distinguish two subcases, depending on whether inequality~\eqref{eqn:nonprop_cond} of Lemma~\ref{lemma:altc} holds or not.

    The first subcase we examine is when the inequality holds.
    Consider the undirected graphs $G_{h,k}$ in Figure~\ref{fig:inf_seq_graph}.
    The linear centralities of $x$ and $y$ with respect to $\vc a$ and $\vc b$ in this graph are
    \begin{align*}
        \celg{\vc a}{G_{h,k}}(x) & = e_a + a_h s + a_{h+1} p + a_{k-1} q + a_k t \\
        \celg{\vc b}{G_{h,k}}(x) & = e_b + b_h s + b_{h+1} p + b_{k-1} q + b_k t \\
        \celg{\vc a}{G_{h,k}}(y) & = e_a + a_h t + a_{h+1} q + a_{k-1} p + a_k s \\
        \celg{\vc b}{G_{h,k}}(y) & = e_b + b_h t + b_{h+1} q + b_{k-1} p + b_k s,
    \end{align*}
    where $e_{a} = a_1 + a_2 + \dots + a_{k+h-2}$ and $e_{b} = b_1 + b_2 + \dots + b_{k+h-2}$.
    Then, the linear centralities $\celg{a}{G_{h,k}}$ and $\celg{b}{G_{h,k}}$ do not agree on $G_{h,k}$ if there is an integer nonnegative solution $(s,p,q,t)$ to the following system of inequalities:
    \begin{align*}
    \begin{cases}
         e_a + a_h s + a_{h+1} p + a_{k-1} q + a_k t \geq e_a + a_h t + a_{h+1} q + a_{k-1} p + a_k s \\
        e_b + b_h s + b_{h+1} p + b_{k-1} q + b_k t < e_b + b_h t + b_{h+1} q + b_{k-1} p + b_k s,
    \end{cases}
    \end{align*}
    which is equivalent to
    \begin{align}\label{eq:system}
    \begin{cases}
        \alpha (s-t) + \alpha' (p-q) \geq 0 \\
        \beta (s-t) + \beta' (p-q) < 0,
    \end{cases}
    \end{align}
    where we have set
    %\begin{align*}
        $\alpha  = a_h - a_k$,
        $\beta  = b_h - b_k$,
        $\alpha' = a_{h+1} - a_{k-1}$ and
        $\beta' = b_{h+1} - b_{k-1}$.
    %\end{align*}
    The two inequalities in~\eqref{eq:system} define two halfplanes (in the coordinate system $X=s-t$ and $Y=p-q$), whose intersection is unbounded and nonempty if $\alpha \beta' \neq \alpha' \beta$, which is exactly inequality~\eqref{eqn:nonprop_cond} of Lemma~\ref{lemma:altc}.
    %\eqref{eq:prop_cond}.
    As a result, the system~\eqref{eq:system} has an integral solution $(X,Y)$, whence nonnegative $s,t,p,q$ can be always found.

    \medskip
    We are left to consider the subcase when inequality~\eqref{eqn:nonprop_cond} does not hold.
    In such a case, we can apply the same line of reasoning as above to the graphs $\widehat G_{h,k}$ in Figure~\ref{fig:inf_seq_graphk}, obtaining the following system of inequalities:
    \begin{align*}
        \begin{cases}
            e_a + a_h s + a_k p \geq e_a + a_h t + a_k q \\
            e_b + b_h s + b_k p < e_b + b_h t + b_k q,
        \end{cases}
    \end{align*}
    where $e_{a} = a_1 + a_2 + \dots + a_{k-1}$ and $e_{b} = b_1 + b_2 + \dots + b_{k-1}$.
    This is equivalent to
    \begin{align*}
        \begin{cases}
            a_h (s-t) - a_k (p-q) \geq 0 \\
            b_h (s-t) - b_k (p-q) < 0.
        \end{cases}
    \end{align*}
    This system has integer nonnegative solution $(s,t,p,q)$ when $a_h b_k \neq b_h a_k$, which holds because of Lemma~\ref{lemma:altc}.
\qed\end{proof}

It is a noteworthy fact that all the graphs used to prove the second item of the statement above are undirected, and connected as long as $a_{h+1} \neq a_{k-1}$ and $h+1<k$ (except for case (a)).

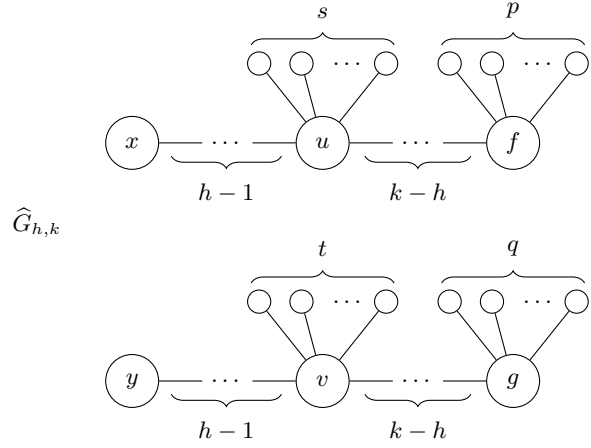
\begin{figure}[htbp]
    \centering
    \begin{tikzpicture}[main/.style = {draw, circle, minimum size=7mm}, node distance=1cm, scale=.7]

        \tikzmath{\d1=2.0; \d2=1.5; \s1=1.0; \d3=-4.5;}

        \node[] (G3) at (-9, \d2-1.5) {$\widehat G_{h,k}$};
        \node[main] (x1) at (\d1*-3.6,\d2) {$x$};
        \node (xh1) at (\d1*-2.7,\d2) {$\dots$};
        \node[main] (h1) at (\d1*-1.8,\d2) {$u$};
        \node (hf1) at (\d1*-0.9,\d2) {$\dots$};
        \node[main] (f1) at (\d1*0,\d2) {{\small$f$}};

        \draw (x1) -- (xh1);
        \draw (xh1) -- (h1);
        \draw (h1) -- (hf1);
        \draw (hf1) -- (f1);

        \node[main, minimum size=3mm] (s1) at (\d1*-1.8-\s1*1.2,\d2*2) {};
        \node[main, minimum size=3mm] (s2) at (\d1*-1.8-\s1*0.4,\d2*2) {};
        \node (s3) at (\d1*-1.8+\s1*0.5,\d2*2) {$\dots$};
        \node[main, minimum size=3mm] (s4) at (\d1*-1.8+\s1*1.2,\d2*2) {};

        \draw (h1) -- (s1);
        \draw (h1) -- (s2);
        \draw (h1) -- (s4);

        \draw [decoration = {raise=-3pt,brace,amplitude=0.5em},decorate,black] (-0.15+\d1*-1.8-\s1*1.2,\d2*2+0.5) -- (+0.15+\d1*-1.8+\s1*1.2,\d2*2+0.5) node[pos=0.5,above=4pt,black] {$s$};

        \node[main, minimum size=3mm] (p1) at (\d1*0-\s1*1.2,\d2*2) {};
        \node[main, minimum size=3mm] (p2) at (\d1*0-\s1*0.4,\d2*2) {};
        \node (p3) at (\d1*0+\s1*0.5,\d2*2) {$\dots$};
        \node[main, minimum size=3mm] (p4) at (\d1*0+\s1*1.2,\d2*2) {};

        \draw (f1) -- (p1);
        \draw (f1) -- (p2);
        \draw (f1) -- (p4);

        \draw [decoration = {raise=-3pt,brace,amplitude=0.5em},decorate,black] (-0.15+\d1*0-\s1*1.2,\d2*2+0.5) -- (+0.15+\d1*0+\s1*1.2,\d2*2+0.5) node[pos=0.5,above=4pt,black] {$p$};
        \draw [decoration = {raise=-3pt,brace,mirror,amplitude=0.5em},decorate,black] (\d1*-3.23,\d2-0.4) -- (\d1*-2.19,\d2-0.4) node[pos=0.5,below=4pt,black] {$h-1$};
        \draw [decoration = {raise=-3pt,brace,amplitude=0.5em},decorate,black] (\d1*-0.4,\d2-0.4) -- (\d1*-1.4,\d2-0.4) node[pos=0.5,below=4pt,black] {$k-h$};

        % second graph

        \node[main] (x2) at (\d1*-3.6,\d2+\d3) {$y$};
        \node (xh2) at (\d1*-2.7,\d2+\d3) {$\dots$};
        \node[main] (h2) at (\d1*-1.8,\d2+\d3) {$v$};
        \node (hf2) at (\d1*-0.9,\d2+\d3) {$\dots$};
        \node[main] (f2) at (\d1*0,\d2+\d3) {{\small$g$}};

        \draw (x2) -- (xh2);
        \draw (xh2) -- (h2);
        \draw (h2) -- (hf2);
        \draw (hf2) -- (f2);

        \node[main, minimum size=3mm] (s12) at (\d1*-1.8-\s1*1.2,\d2*2+\d3) {};
        \node[main, minimum size=3mm] (s22) at (\d1*-1.8-\s1*0.4,\d2*2+\d3) {};
        \node (s32) at (\d1*-1.8+\s1*0.5,\d2*2+\d3) {$\dots$};
        \node[main, minimum size=3mm] (s42) at (\d1*-1.8+\s1*1.2,\d2*2+\d3) {};

        \draw (h2) -- (s12);
        \draw (h2) -- (s22);
        \draw (h2) -- (s42);

        \draw [decoration = {raise=-3pt,brace,amplitude=0.5em},decorate,black] (-0.15+\d1*-1.8-\s1*1.2,\d2*2+0.5+\d3) -- (+0.15+\d1*-1.8+\s1*1.2,\d2*2+0.5+\d3) node[pos=0.5,above=4pt,black] {$t$};

        \node[main, minimum size=3mm] (p12) at (\d1*0-\s1*1.2,\d2*2+\d3) {};
        \node[main, minimum size=3mm] (p22) at (\d1*0-\s1*0.4,\d2*2+\d3) {};
        \node (p32) at (\d1*0+\s1*0.5,\d2*2+\d3) {$\dots$};
        \node[main, minimum size=3mm] (p42) at (\d1*0+\s1*1.2,\d2*2+\d3) {};

        \draw (f2) -- (p12);
        \draw (f2) -- (p22);
        \draw (f2) -- (p42);

        \draw [decoration = {raise=-3pt,brace,amplitude=0.5em},decorate,black] (-0.15+\d1*0-\s1*1.2,\d2*2+0.5+\d3) -- (+0.15+\d1*0+\s1*1.2,\d2*2+0.5+\d3) node[pos=0.5,above=4pt,black] {$q$};
        
        \draw [decoration = {raise=-3pt,brace,mirror,amplitude=0.5em},decorate,black] (\d1*-3.23,\d2-0.4+\d3) -- (\d1*-2.19,\d2-0.4+\d3) node[pos=0.5,below=4pt,black] {$h-1$};
        \draw [decoration = {raise=-3pt,brace,amplitude=0.5em},decorate,black] (\d1*-0.4,\d2-0.4+\d3) -- (\d1*-1.4,\d2-0.4+\d3) node[pos=0.5,below=4pt,black] {$k-h$};
    \end{tikzpicture}
    \caption{The (family of) graphs $\widehat G_{h,k}$ used to distinguish non-proportional sequences when $a_{h+1} = a_{k-1}$ or $h+1 = k$. The horizontal dotted lines at the bottom of the two connected components represent paths of length $h-1$ (and $k-h$, respectively). Note that when $h=1$, $u$ gets to coincide with $x$, hence the $s$ nodes in such a case are connected to $x$ (everything holds symmetrically for $y$).}\label{fig:inf_seq_graphk}
\end{figure}

As an example, consider the linear centralities $\celg{a}{}$ and $\celg{b}{}$, with
\[
    \vc a=\left(0,1,\frac{1}{2},\frac{1}{3},\dots\right)^T \quad\text{and}\quad \vc b=\left(0,\frac{1}{2},\frac{1}{4},\frac{1}{8}\dots\right)^T,
\]
i.e., harmonic centrality and exponential decay centrality with $\delta=1/2$.
These two centralities can be distinguished on the graph $\widehat{G}_{1,3}$ with $s=q=0$, $p=6$, and $t=2$.
In fact, we have:
\begin{align*}
    \celg{\vc a}{\widehat{G}_{1,3}}(x) = \frac{7}{2} & \geq \frac{7}{2} = \celg{\vc a}{\widehat{G}_{1,3}}(y) \\
    \celg{\vc b}{\widehat{G}_{1,3}}(x) = \frac{3}{2} & < \frac{7}{4} = \celg{\vc b}{\widehat{G}_{1,3}}(y).
\end{align*}

\section{Graphs with Many Representable Permutations}

The previous section aimed at discussing when and on which graphs two sequences of coefficients yield different centralities. We now ask a dual question: given a graph $G$, which permutations of its nodes can be induced by linear centralities? In principle, we may have graphs where all linear centralities agree, regardless of how you choose the coefficients: the fact that there is a graph on which each pair of non-proportional centralities can be distinguished does not rule out the possibility that \emph{on some graph} they may all agree.

On the other hand, it may be possible to have graphs that allow for many different permutations of their nodes, depending on the choice of coefficients. To avoid the complications introduced by nodes having the same centrality value, let us stick to the case with no ties.

\begin{definition}[Representable permutation]
We say that a permutation $\pi \in S_n$ is \emph{representable} by a graph $G$
with $n$ nodes iff there is a choice of coefficients $\vc a \in \RR^\NN$ such that $\cel{a}$ has no ties and induces the permutation $\pi$ on $G$.
\end{definition}

Of course, we may classify graphs depending on how many permutations they can represent. More precisely, let us define the \emph{representativeness} of a graph $G$ with $n$ nodes as follows:
\[
    R(G)=\frac{|\{\pi \in S_n \mid \text{$\pi$ is representable by $G$}\}|}{n!} \leq 1.
\]
This ratio quantifies how expressive linear centralities are over a rigid graph $G$. If $R(G)=1$, linear centralities cannot be less expressible than any other centrality: by choosing coefficients suitably, we can induce any permutation of the nodes, hence achieving what any other centrality can achieve.
On the other hand, of course, if $R(G)<1$ there are permutations that cannot be obtained using a linear centrality, although they may be obtained by some other centrality. 
In the rest of this section we show that at least in some cases linear centralities are maximally expressive (i.e., they have maximum representativeness); if we require the graph to be strongly connected, we will be able to obtain the same result but only asymptotically.

From now on, we will drop the assumption that $a_0=0$, as it will play a role in quantifying representativeness.

\subsection{Rouché-Capelli Theorem for the Non-Connected Case}

Given a graph $G$ of $n$ nodes, and a vector $\vc v \in \RR^n$, determining if there exists a coefficient vector $\vc a \in \RR^\NN$ such that $\cel{a}(i)=v_i$ for all $i \in [n]$ can be seen as trying to find a solution $\vc a$ for the following linear system of equations:
\begin{equation}
\label{eqn:finda}
    C_G \cdot \vc a = \vc v.
\end{equation}
This is a system of $n$ linear equalities with $n$ unknowns, hence by the Rouché-Capelli theorem~\cite{shafarevich2012linear}, it has a solution if and only if 
\begin{equation}
\label{eqn:rouche-capelli}
    \Rk{C_G}=\Rk{C_G|\vc v},
\end{equation}
where $\Rk A$ is the rank of matrix $A$, and $A|\vc x$ is the matrix $A$ augmented with the  column $\vc x$.

Now, if $\Rk{C_G}=n$, equation~\eqref{eqn:rouche-capelli} always holds, so if we find a graph $G$ whose distance-count matrix has full rank, then the graph has representativeness equal to $1$. Much more than this, in fact: not only can we find linear centralities that rank the nodes in any order we please, but we may even choose the \emph{actual centrality values} for each single node. Such a graph exists:

\begin{theorem}
    For every $n \geq 3$, the graph $G_n$ in Figure~\ref{fig:discon_graph} has representativeness $R(G_n)=1$.
\end{theorem}
\begin{proof}
The graph in Figure~\ref{fig:discon_graph} has the following distance-count matrix
{
\begin{equation*}
    C = 
    \setlength{\arraycolsep}{3pt}
    \renewcommand\arraystretch{0.9}
    \begin{bmatrix}
        1 & 0 & 0 & 0 & 0 & \dots & 0 & 0 & 0 \\
        1 & n-1 & 0 & 0 & 0 & \dots & 0 & 0 & 0 \\
        1 & 1 & n-2 & 0 & 0 & \dots & 0 & 0 & 0 \\
        1 & 1 & 1 & n-3 & 0 & \dots & 0 & 0 & 0 \\
        \vdots & \vdots & \vdots & \vdots & \vdots & \ddots & \vdots & \vdots & \vdots \\
        1 & 1 & 1 & 1 & 1 & \dots & 1 & 2 & 0 \\
        1 & 1 & 1 & 1 & 1 & \dots & 1 & 1 & 1 \\
    \end{bmatrix}
\end{equation*}
}
which is easily seen to have rank $n$ (it is lower-triangular with no zeros on its diagonal).
\qed\end{proof}

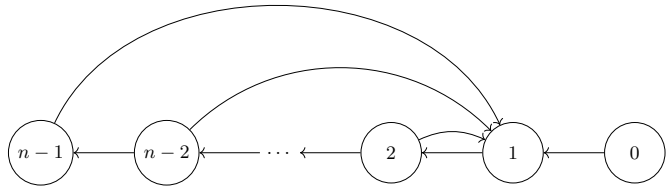
\begin{figure}[h]
    \centering
    \scalebox{0.8}{
    \begin{tikzpicture}[main/.style = {draw, circle, minimum size=10mm}, node distance=1cm]

        \node[main] (x1) {$0$};
        \node[main] (x2) [left=of x1] {$1$};
        \node[main] (x3) [left=of x2] {$2$};
        \node (xk) [left=of x3] {$\dots$};
        \node[main] (xn-1) [left=of xk] {$n-2$};
        \node[main] (xn) [left=of xn-1] {$n-1$};

        \path[->]
        (x1) edge (x2)
        (x2) edge (x3)
        (x3) edge (xk)
        (xk) edge (xn-1)
        (xn-1) edge (xn)
        (x3) edge[bend angle=25,bend left] (x2)
        (xn-1) edge[bend angle=45,bend left] (x2)
        (xn) edge[bend angle=65,bend left] (x2);
    \end{tikzpicture}}
    \caption{A graph $G_n$ with representativeness $R(G_n)=1$.}\label{fig:discon_graph}
\end{figure}

\subsection{Farkas' Lemma for the General Case}

The main drawback of the solution in Figure~\ref{fig:discon_graph} is that the graph $G_n$ is not strongly connected, albeit almost so.
Therefore, we still wonder if a \emph{strongly connected} graph can be permuted in all possible ways by linear centralities.

To tackle this problem (or, more in general, to have a tool to decide if a given permutation is representable by a graph), we shall use the following well-known result~\cite{farkas1902theorie,gale1951linear}:

\begin{lemma}[Farkas' Lemma]
\label{lemma:farkas}
    Given $A \in \RR^{m \times n}$ and $\vc b \in \RR^m$, the following statements are equivalent:
    \begin{enumerate}
        \item there exists $\vc x \in \RR^n$ such that $A\vc x \leq \vc b$;
        \item there exists \emph{no} $\vc y \in \RR^m$ such that $A^T\vc y=\vc 0$, $\vc y \geq \vc 0$ and $\vc b^T\vc y<\vc 0$.
    \end{enumerate}
\end{lemma}

A special version of Farkas' Lemma is the following:
\begin{corollary}
\label{cor:farkas-orig}
    Given $A \in \RR^{m \times n}$, the following statements are equivalent:
    \begin{enumerate}
        \item there exists $\vc x \in \RR^n$ such that $A\vc x < \vc 0$;
        \item there exists \emph{no} $\vc y \in \RR^m$ such that $A^T\vc y=\vc 0$, $\vc y \geq \vc 0$ and $\vc y \neq \vc 0$.
    \end{enumerate}
\end{corollary}

    \begin{proof}
    For every $\varepsilon>0$, let ($1_\varepsilon$) and ($2_\varepsilon$) be the two equivalent statements of Lemma~\ref{lemma:farkas} when $\vc b=-\varepsilon \vc 1$. Clearly, the first statement in the corollary is equivalent to $\exists \varepsilon >0$ such that ($1_\varepsilon$). By Farkas' Lemma, the latter is equivalent to $\exists \varepsilon >0$ such that ($2_\varepsilon$) holds, that is:
    \[
        \exists \varepsilon>0 \; \nexists \vc y\in \RR^m \text{ such that } A^T\vc y=\vc 0, \quad \vc y \geq \vc 0, \quad -\varepsilon \vc 1^T \vc y<0.
    \]
    Note that if $\varepsilon>0$ and $\vc y\geq \vc 0$ then $-\varepsilon \vc 1^T \vc y<\vc 0$ unless $\vc y=\vc 0$. So we can express the above condition as:
    \[
        \exists \varepsilon>0 \; \nexists \vc y\in \RR^m \text{ such that } A^T\vc y=\vc 0, \quad \vc y \geq \vc 0, \quad \vc y\neq \vc 0,
    \]
    or simply (since the latter condition does not depend on $\varepsilon$ anymore)
    \[
        \nexists \vc y\in \RR^m \text{ such that } A^T\vc y=\vc 0, \quad \vc y \geq \vc 0, \quad \vc y\neq \vc 0.
    \]
    \qed\end{proof}

Define the matrix $\Xi \in \RR^{(n-1)\times n}$ as follows
{
\begin{equation*}
    \Xi = 
    \setlength{\arraycolsep}{3pt}
    \renewcommand\arraystretch{0.9}
    \begin{bmatrix}
    -1 & 1 & 0 & 0 & \dots & 0 & 0 & 0\\
    0 & -1 & 1 & 0 & \dots & 0 & 0 & 0\\
    0 & 0 & -1 & 1 & \dots & 0 & 0 & 0\\
    \vdots & \vdots & \vdots & \vdots & \ddots & \vdots & \vdots & \vdots \\
    0 & 0 & 0 & 0 & \dots & -1 & 1 & 0\\
    0 & 0 & 0 & 0 & \dots & 0 & -1 & 1
    \end{bmatrix}.
\end{equation*}
}

Then, we can use Corollary~\ref{cor:farkas-orig} to determine if a given graph can represent a permutation as follows:

\begin{theorem}\label{thm:reprfark}
    A permutation $\pi$ is representable by a graph $G$ if and only if there does not exist $\vc y\in \RR^{n-1}$ such that 
    \begin{equation}
    \label{eqn:farkas}
        \left(\Xi R_\pi C_G\right)^T \vc y=0
    \end{equation}
    with $\vc y\geq \vc 0$ and $\vc y \neq \vc 0$, where $R_\pi$ is the permutation matrix of $\pi$ (see Section~\ref{sec:notations}).
\end{theorem}
\begin{proof}
    For every $\vc v \in \RR^n$, the vector is monotonically decreasing (i.e., $v_0>v_1>v_2>\dots>v_{n-1}$) if and only if $\Xi \vc v < \vc 0$, because the first inequality is $-v_0+v_1<0$, that is, $v_0>v_1$, and so on.

    Recall that $\pi$ is representable by $G$ iff there exists $\vc a \in \RR^n$ such that $\cel{\vc a}$ induces $\pi$ on $G$.
    This means that
    \[ \cel{\vc a}(\pi(i)) > \cel{\vc a}(\pi(i+1)) \]
    for all $i \in [n-1]$.
    In other words, $\pi$ is representable by $G$ iff there exists $\vc a \in \RR^n$ such that for all $i \in [n-1]$
    \[
        \left(C_G \cdot \vc a\right)_{\pi(i)} > \left(C_G \cdot \vc a\right)_{\pi(i+1)},
    \]
    or, equivalently,
    \[
        \left(R_\pi C_G \cdot \vc a\right)_i > \left(R_\pi C_G \cdot \vc a\right)_{i+1}.
    \]
    Applying the observation above to $\vc v=R_\pi C_G \cdot \vc a$, this is equivalent to saying that there exists $\vc a \in \RR^n$ such that
    \[
        \Xi R_\pi C_G \cdot \vc a < 0.
    \]
    The result is obtained by applying Corollary~\ref{cor:farkas-orig} with $A=\Xi R_\pi C_G$.\qed
\end{proof}

In other words, to see if $\pi$ is representable by $G$ we can look at the system~\eqref{eqn:farkas} of $n$ equations with $n-1$ indeterminates: $\pi$ is representable if and only if no non-negative non-trivial solution exists for the system.

More precisely, we will be using the following:
\begin{corollary}\label{cor:farkas}
    For a given graph $G$ and permutation $\pi$,
    let $X \in \RR^{n \times k}$ be any matrix whose columns are linear combinations of the columns of $C_G$. Let $\rho=\pi^{-1}$.
    Then every $\vc y \in \RR^{n-1}$ satisfying~\eqref{eqn:farkas} also satisfies
    \begin{equation}
        \label{eqn:farkas-ter}
        \sum_{h:\ 0<\rho(h)} x_{h,i}y_{\rho(h)-1}  - \sum_{h:\ \rho(h)<n-1} x_{h,i}y_{\rho(h)} = 0 \qquad \forall i \in [k].
    \end{equation}
\end{corollary}
\begin{proof}
Let us rewrite the $i$-th element of the vector on the LHS of the system~\eqref{eqn:farkas} in a more intelligible form:
\begin{align*}
    \Big( (\Xi R_\pi C)^T \vc y\Big)_i =
    \sum_{j<n-1} \left(\Xi R_\pi C\right)^T_{i,j} y_j & =
    \sum_{j<n-1} \left(\Xi R_\pi C\right)_{j,i} y_j = \\
    & =\sum_{j<n-1}\sum_{h_1,h_2} \Xi_{j,h_1} (R_\pi)_{h_1,h_2} c_{h_2,i} y_j
\end{align*}
where the variables $h_1,h_2$ range over $[n]$. Now, recalling the definition of $\Xi$ and $R_\pi$, we can rewrite~\eqref{eqn:farkas} as:
\begin{equation*}
    %\label{eqn:farkas-bis}
    \sum_{j<n-1} \left(c_{\pi(j+1),i}-c_{\pi(j),i}\right) y_j = 0 \qquad \forall i \in [n].
\end{equation*}
If we let $\rho=\pi^{-1}$, we can re-factor these equations as:
\begin{equation}
    \label{eqn:farkas-ter-pre}
    \sum_{h:\ 0<\rho(h)} c_{h,i}y_{\rho(h)-1}  - \sum_{h:\ \rho(h)<n-1} c_{h,i}y_{\rho(h)} = 0 \qquad \forall i \in [n].
\end{equation}
So $\vc y$ satisfies each of the $n$ equalities~\eqref{eqn:farkas-ter-pre}. There is one such equality for every column of $C_G$. If all such equalities are satisfied, then \emph{a fortiori} all their linear combinations are also satisfied.\qed
\end{proof}

A useful observation is that, if we look at a column $i$ of a matrix $X$ having only one single non-zero entry $x_{h,i}$, then condition~\eqref{eqn:farkas-ter} becomes
\[
    y_{\rho(h)}=y_{\rho(h)-1}
\]
if $0<\rho(h)<n-1$; if, instead, $\rho(h)=0$ ($\rho(h)=n-1$, respectively) then the condition implies $y_0=0$ ($y_{n-2}=0$, resp.).

\medskip
The idea of representable permutations, besides being useful to estimate expressivity of linear centralities on a given graph, has another opposite application that we may call \emph{robustness}. Suppose that you know that an adversary will rank the nodes of your graph using \emph{some unknown} linear centrality (it can be closeness, harmonic, or anything else), and suppose that you want to be sure, say, that a certain node (or set of nodes) is ranked higher than another node (or set of nodes). If you can fiddle a bit with the arcs of the graph, you may want to try to impose a structure which guarantees that only the permutations that you like are representable, thus ruling out the possibility that the event you dislike happens, \emph{whatever linear centrality will be applied by the adversary}. The usage of equation~\eqref{eqn:farkas-ter} is a very powerful tool to guarantee this kind of properties.

\subsection{Representativeness of Strongly Connected Graphs}
\label{sec:repstr}

As a first attempt at providing a sequence of strongly connected graphs with large representativeness, we have the following result:

\begin{theorem}
    \label{thm:rhomid}
    For $n>3$, consider the graph $G_n$ in Figure~\ref{fig:mdc_example}, and let $\pi \in S_n$ and $\rho=\pi^{-1}$. The following two statements are equivalent:
    \begin{enumerate}
        \item $\rho(0)<\rho(n-1)<\rho(1)$ or $\rho(1)<\rho(n-1)<\rho(0)$;
        \item the system~\eqref{eqn:farkas} has no non-negative non-trivial solution.
    \end{enumerate}
\end{theorem}
\begin{proof} 
The proof is rather long, hence we prefer to place it in the appendix.
\qed
\end{proof}

The set of permutations satisfying the statement of Lemma~\ref{thm:rhomid} has cardinality $n!/3$, as shown by:

\begin{lemma}
    If $n\geq3$, there are $n!/3$ permutations $\rho \in S_n$, such that $\rho(0)<\rho(n-1)<\rho(1)$ or $\rho(1)<\rho(n-1)<\rho(0)$.
\end{lemma}
\begin{proof}
    We prove this result by induction on $n$. For $n=3$, there are only two permutations satisfying the constraint, and $3!/3=2$.
    For the inductive step, take any permutation $\rho'\in S_{n-1}$ satisfying the conditions in the statement, and build another permutation $\rho\in S_{n}$ so that:
    \begin{itemize}
        \item the relative order of $\rho(i)$ and $\rho(j)$ is the same as that of $\rho'(i)$ and $\rho'(j)$, for $i,j \in [n-2]$;
        \item the relative order of $\rho(i)$ and $\rho(n-1)$ is the same as that of $\rho'(i)$ and $\rho'(n-2)$, for $i \in [n-2]$.
    \end{itemize}
    Naturally $\rho(n-2)$ can be defined in $n$ ways (because it is not involved in any of the above conditions). By inductive hypothesis, there are $(n-1)!/3$ permutations $\rho'\in S_{n-1}$ satisfying the hypothesis; we have just built the (only) possible $n(n-1)!/3=n!/3$ permutations $\rho\in S_{n}$ with the same property.
\qed\end{proof}

%It is easy to see (by induction on $n$) that there are $n!/3$ permutations $\pi$ satisfying the condition of the above theorem. 

From what we have shown so far, we can conclude that:

\begin{corollary}
    For every $n>3$, the graph $G_n$ in Figure~\ref{fig:mdc_example} is strongly connected and $R(G_n)=1/3$.
\end{corollary}

\subsection{A graph that represents almost all permutations}

In this section, we assume that $n \ge 4$, and we let $\PP_n$ be the set of permutations $\pi\in S_n$, with $\rho=\pi^{-1}$ such that one of the following (mutually exclusive) conditions holds:
\begin{enumerate}
    \item\label{enu:ppone} $2 \le \rho(n-1) \le n-3$;
    \item\label{enu:pptwo} $\rho(n-1) = 1 \land \rho(n-2) \neq 0$;
    \item\label{enu:ppthree} $\rho(n-1) = n-2 \land \rho(n-2) \neq n-1$.
\end{enumerate}

The following preliminary statement shows that almost all permutations belong to this set (in an asymptotic sense):
\begin{lemma}
    The set $\PP_n$ has cardinality $n!(n-3)/(n-1)$.
\end{lemma}
\begin{proof}
    Note that in $\PP_n$ there is no permutation such that $n-1$ is in the first (i.e., $\rho(n-1)=0$) or last (i.e., $\rho(n-1)=n-1$) position.

    We start with the set of permutations satisfying~\eqref{enu:ppone}:
    it is easy to see that these are the permutations where $n-1$ is in one of the $(n-4)$ \q{middle} positions of the permutation, and the other $n-1$ elements can be permuted freely.
    Therefore, there are in total $(n-4)(n-1)!$ such permutations.

    For the second constraint, we are considering the $(n-1)!$ permutations where $n-1$ is in the second position, and from these we have to cut off the $(n-2)!$ permutations of $(n-1)$ elements where $n-2$ comes first.
    Hence, there are $(n-1)!-(n-2)!$ permutations satisfying~\eqref{enu:pptwo}.
    The third constraint is analogous.

    Hence, the total number of permutations in $\PP_n$ is:
    \begin{align*}
        |\PP_n| & = (n-4)(n-1)! + 2 \big( (n-1)! - (n-2)! \big) = (n-1)! \Big( n-2-\frac{2}{n-1} \Big) \\
            & = (n-1)! \; \frac{n^2 - 3n}{n-1} = n!\;\frac{n-3}{n-1}.
    \end{align*}
\qed\end{proof}

\begin{figure}
    \begin{minipage}[c]{1\linewidth}
    \centering
    \scalebox{0.8}{
    \begin{tikzpicture}[main/.style = {draw, circle, minimum size=10mm}, node distance=0.8cm]
        \node[main] (x0) {$n-3$};
        \node[main] (x1) [left=of x0] {$0$};
        \node[main] (x2) [left=of x1] {$1$};
        \node (dots) [left=of x2] {$\dots$};
        \node[main] (x3) [left=of dots] {$n-4$};
        \node[main] (x4) [left=of x3] {$n-1$};
        \node[main] (x5) [left=of x4] {$n-2$};

        \path[->]
        (x0) edge (x1)
        (x1) edge (x2)
        (x2) edge (dots)
        (dots) edge (x3)
        (x3) edge (x4)
        (x4) edge(x5)

        (x3) edge[bend angle=25, bend right] (x0)
        (x5) edge[bend angle=20, bend right] (x4)
        
        (x2) edge[bend angle=15, bend left] (x1)
        (x3) edge[bend angle=35, bend left] (x1)
        (x4) edge[bend angle=55, bend left] (x1)
        (x5) edge[bend angle=75, bend left] (x1);
    \end{tikzpicture}}
    \end{minipage}\\\vspace{10pt}
    \begin{minipage}[c]{1\linewidth}
    \centering
    $C_{G'_n} = 
    \setlength{\arraycolsep}{3pt}
    \renewcommand\arraystretch{0.9}
    \begin{bmatrix}
        1 & n-1 & 0 & 0 & 0 & \dots & 0 & 0 & 0 \\
        1 & 1 & n-2 & 0 & 0 & \dots & 0 & 0 & 0 \\
        1 & 1 & 1 & n-3 & 0 & \dots & 0 & 0 & 0 \\
        \vdots & \vdots & \vdots & \vdots & \vdots & \ddots & \vdots & \vdots & \vdots \\
        1 & 1 & 1 & 1 & 1 & \dots & 1 & 2 & 0 \\
        1 & 1 & 1 & 1 & 1 & \dots & 1 & 1 & 1 \\
        1 & 2 & 1 & 1 & 1 & \dots & 1 & 1 & 0 \\
    \end{bmatrix}$
    \end{minipage}
    \caption{\label{fig:conn_graph}The graph $G'_n$ with $R(G'_n)=(n-3)/(n-1)$ and its distance-count matrix $C_{G'_n}$.}
\end{figure}
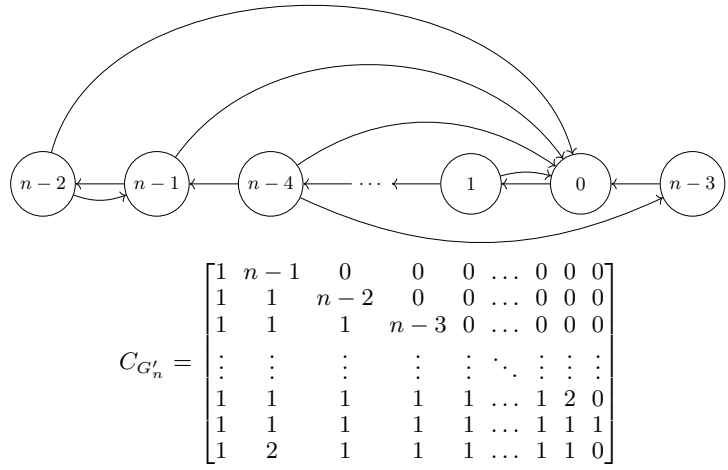

The next result connects $\PP_n$ to the graph $G'_n$ of Figure~\ref{fig:conn_graph}:

\begin{theorem}
    \label{thm:hardfarkas}
    Consider the graph $G'_n$ of Figure~\ref{fig:conn_graph} and let $\pi \in S_n$. The following two statements are equivalent:
    \begin{enumerate}
        \item\label{enu:hardfarkasone} $\pi$ belongs to $\PP_n$;
        \item\label{enu:hardfarkastwo} the permutation $\pi$ is representable by the graph $G'_n$.
    \end{enumerate}
\end{theorem}
\begin{proof}
We start proving that~\eqref{enu:hardfarkasone}$\implies$\eqref{enu:hardfarkastwo}. Let $\pi \in \PP_n$ be fixed and $\vc y \in \RR^{n-1}$ be a non-negative solution of the system~\eqref{eqn:farkas}: we will prove that $\vc y=\vc 0$. Using Corollary~\ref{cor:farkas}, the vector $\vc y$ must satisfy:
\begin{equation}
    \label{eqn:farkas-terapplied}
    \sum_{h:\ 0<\rho(h)} x_{h,i}y_{\rho(h)-1}  - \sum_{h:\ \rho(h)<n-1} x_{h,i}y_{\rho(h)} = 0 \qquad \forall i, 
\end{equation}
where $X$ represents either the distance-count matrix of $G'_n$ or, more generally, any matrix obtained by linearly combining the columns of the distance-count matrix. Index $i$ ranges over all the columns of $X$, and there is one equation per column. 

It is convenient to extend $\vc y$ adding one element to the left ($y_{-1}$) and one element to the right ($y_{n-1}$), with $y_{-1}=y_{n-1}=0$. This way, we can rewrite~\eqref{eqn:farkas-terapplied} in the simpler form
\begin{equation}
    \label{eqn:farkas-terappliedbis}
    \sum_{h=0}^{n-1} x_{h,i} \left(y_{\rho(h)-1}  - y_{\rho(h)}\right) = 0 \qquad \forall i.
\end{equation}

Instead of using directly the matrix $X=C_{G'_n}$, we prefer to use the following:
\[
    X = 
    \setlength{\arraycolsep}{3pt}
    \renewcommand\arraystretch{0.9}
    \begin{pNiceMatrix}
        n-1 &\Block[fill=black!25]{9-5}{} 0 & 0 & 0 & \dots & 0 & 0 & 0 \\
        3-n & n-2 & 0 & 0 & \dots & 0 & 0 & 0 \\
        0 & 4-n & n-3 & 0 & \dots & 0 & 0 & 0 \\
        0 & 0 & 5-n & n-4 & \dots & 0 & 0 & 0 \\
        \vdots & \vdots & \vdots & \vdots & \vdots & \ddots & \vdots & \vdots \\% & \vdots \\
        0 & 0 & 0 & 0 & \dots & 3 & 0 & 0 \\
        0 & 0 & 0 & 0 & \dots & -1 & 2 & 0 \\
        0 & 0 & 0 & 0 & \dots & 0 & 0 & 1 \\
        1 & 0 & 0 & 0 & \dots & 0 & 1 & 0 \\
    \end{pNiceMatrix}.
\]

This matrix has $n-1$ columns (i.e., it yields $n-1$ equations), and its columns are obtained  from $C_{G'_n}$ (Figure~\ref{fig:conn_graph}) starting from the one of index $1$, subtracting from each column the following one, except the last one that is left untouched. 

\paragraph{Inner columns.}
All the columns of $X$ except the first one and the last two contain exactly two non-zero entries, which are consecutive and with opposite sign. We call these ($n-4$) columns the ``inner columns'', highlighted in grey in the matrix above.
Equation~\eqref{eqn:farkas-terappliedbis}, when applied to these columns, becomes
\[    (n-i-1) \left(y_{\rho(i)-1} - y_{\rho(i)}\right) = (n-i-3)\left(y_{\rho(i+1)-1} - y_{\rho(i+1)}\right)
\]
for $i=1,\dots,n-4$, which implies that:
\[
    \sgn\left(y_{\rho(i)-1} - y_{\rho(i)}\right) = \sgn\left(y_{\rho(i+1)-1} - y_{\rho(i+1)}\right).    
\]
This is, in fact, a chain of equalities between signs: 
\begin{multline}
    \label{eqn:sign}
    \sgn\left(y_{\rho(1)-1} - y_{\rho(1)}\right)=\sgn\left(y_{\rho(2)-1} - y_{\rho(2)}\right)=\\
    =\sgn\left(y_{\rho(3)-1} - y_{\rho(3)}\right) = \dots = 
    \sgn\left(y_{\rho(n-3)-1} - y_{\rho(n-3)}\right).
\end{multline}

\paragraph{The last two columns.}
The very last column of $X$ yields the equation:
\begin{equation}
    \label{eqn:lastcolumn}
    y_{\rho(n-2)-1} = y_{\rho(n-2)}.
\end{equation}

The next-to-last column, instead, gives
\[
    2 \left(y_{\rho(n-3)-1} - y_{\rho(n-3)}\right) =  -\left(y_{\rho(n-1)-1} - y_{\rho(n-1)}\right)
\]
which entails
\begin{equation}
    \label{eqn:ntlcolumn}
    \sgn \left(y_{\rho(n-3)-1} - y_{\rho(n-3)}\right) =  -\sgn \left(y_{\rho(n-1)-1} - y_{\rho(n-1)}\right).
\end{equation}
Observe that the left-hand side corresponds to the last element of the chain~\eqref{eqn:sign}. 

\paragraph{The first column.}
Finally, the first column of $X$ yields:
\[   (n-1) \left(y_{\rho(0)-1}-y_{\rho(0)}\right) +
    (3-n) \left(y_{\rho(1)-1}-y_{\rho(1)}\right) +\\
    \left(y_{\rho(n-1)-1}-y_{\rho(n-1)}\right) = 0.    
\]
Since $y_{\rho(1)-1}-y_{\rho(1)}$  is part of the chain~\eqref{eqn:sign} and has a negative coefficient, the second summand has the same sign as the last one because of~\eqref{eqn:ntlcolumn}. So the first addend must have the same sign as the elements of the chain~\eqref{eqn:sign}:
\begin{equation}
    \label{eqn:firstcolumn}
    \sgn( \left(y_{\rho(0)-1}-y_{\rho(0)}\right) = \sgn \left(y_{\rho(1)-1}-y_{\rho(1)}\right). 
\end{equation}

\paragraph{The graph $H$.} To consolidate all the knowledge gained so far, consider the undirected graph $H$ whose $n$ vertices are the pairs $(p-1, p)$ for $p=0,\dots,n-1$, and with an edge between $(p-1, p)$ and $(q-1, q)$ whenever $p=\rho(i)$ and $q=\rho(i+1)$ for some $i \in \{0,\dots,n-4\}$. 
Each vertex $(p-1, p)$ is interpreted as representing the difference $y_{p-1}-y_p$: when talking about a vertex, we may say ``the sign of the vertex $(p-1,p)$'' to mean ``the sign of the difference $y_{p}-y_{p-1}$''.

The presence of an edge between $(p-1, p)$ and $(q-1, q)$ represents the fact that
\[
        \sgn\left(y_{p-1}-y_{p}\right) = \sgn\left(y_{q-1}-y_{q}\right),
\]
which we know by~\eqref{eqn:sign} and~\eqref{eqn:firstcolumn}.

Looking at $H$, we see that it is in fact a chain starting at vertex $(\rho(0)-1, \rho(0))$ and ending at $(\rho(n-3)-1, \rho(n-3))$. This chain includes $n-2$ vertices. All these vertices, hence, have the same sign. We refer to this sign as the ``sign of the chain''.

There are $2$ vertices left out of the chain (i.e., isolated in the graph), precisely the vertices in this set:
\[
    J=\{(\rho(n-2)-1, \rho(n-2)), (\rho(n-1)-1, \rho(n-1))\}.
\]
We have some information about these vertices, though. 
The first one is in fact simple, because it has null sign (i.e., $\sgn(y_{\rho(n-2)-1}-y_{\rho(n-2)})=0)$), as in~\eqref{eqn:lastcolumn}.
The second one instead has a sign \emph{opposite} to the sign of the chain, because of~\eqref{eqn:ntlcolumn}.

A graphical depiction of graph $H$ and of all the information we have gained about its vertices is given in Figure~\ref{fig:proof}.

\begin{figure}
    \centering
    \begin{tikzpicture}[main/.style=draw, rectangle]
        \node[draw, rectangle] (r0) at (0,0) {$y_{\rho(0)-1}-y_{\rho(0)}$};
        \node[draw, rectangle] (r1) at (0,-1) {$y_{\rho(1)-1}-y_{\rho(1)}$};
        \node[draw, rectangle] (r2) at (0,-2) {$y_{\rho(2)-1}-y_{\rho(2)}$};
        \node (rdots) at (0,-3) {\myvdots};
        \node[draw, rectangle] (rn4) at (0,-4) {$y_{\rho(n-4)-1}-y_{\rho(n-4)}$};
        \node[draw, rectangle] (rn3) at (0,-5) {$y_{\rho(n-3)-1}-y_{\rho(n-3)}$};
        \draw (r0) -- (r1);
        \draw (r1) -- (r2);
        \draw (r2) -- (rdots);
        \draw (rdots) -- (rn4);
        \draw (rn4) -- (rn3);
        \node [draw, rectangle, label=right:{$^0$}] (rn2) at (6,-2) {$y_{\rho(n-2)-1}-y_{\rho(n-2)}$};
        \node [draw, rectangle] (rn1) at (6,-3) {$y_{\rho(n-1)-1}-y_{\rho(n-1)}$};
        \draw [decorate, decoration=zigzag] (rn1) -- (rn3);
        \draw (r0) -- (r1);
    \end{tikzpicture}
    \caption{\label{fig:proof}The graph $H$. Node $(p-1, p)$ is labelled explicitly as the difference it represents (i.e., $y_{p-1}-y_p$). Edges should be read ``has-the-same-sign-as''. The zigzag line is read ``has-opposite-sign-than''. The isolated node marked with $0$ is known to have sign equal to zero (equation~\eqref{eqn:lastcolumn}).}
\end{figure}
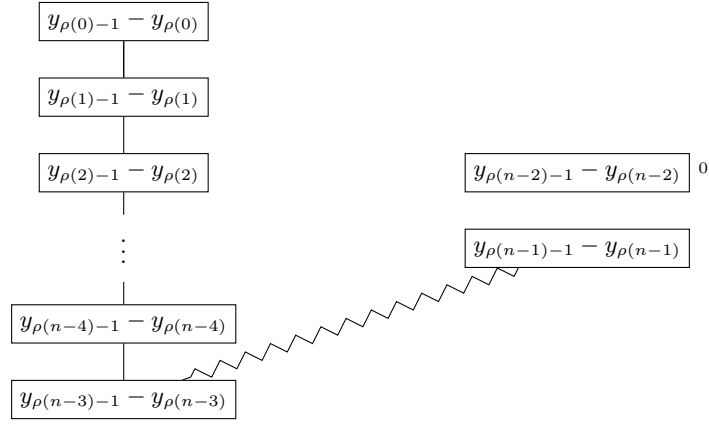

\paragraph{The special vertices.}
Observe that all variables $y_i$ (for $i=0,\dots,n-2$) appear in exactly two vertices, whereas the fake variables $y_{-1}$ and $y_{n-1}$ appear in exactly one vertex. We call the vertices containing $y_{-1}$ and $y_{n-1}$ ``special'' because the two variables $y_{-1}$ and $y_{n-1}$ are synthetic, they are not really part of $\vc y$. So the sign of $(-1,0)$ is $\leq 0$, whereas the sign of $(n-2,n-1)$ is $\geq 0$. 

\paragraph{Case-by-case analysis.} We end the proof with a case-by-case analysis, relying on the fact that $\pi \in \PP_n$.
\begin{enumerate}
    \item \emph{Both special vertices belong to the chain.} In this case, necessarily the sign of the chain is $0$ (all the vertices in the chain have the same sign, but there are two vertices of signs $\geq 0$ and $\leq 0$). So also the sign of node $(\rho(n-1)-1, \rho(n-1))$ (because of~\eqref{eqn:ntlcolumn}) is zero. In other words, all the differences appearing in all vertices are null.
    Also note that for the special vertices $y_{-1}-y_0=0$ implies $y_0=0$, and similarly $y_{n-2}=0$. Using all the other equalities, we obtain $\vc y=\vc 0$.
    \item \emph{Both special vertices are out of the chain.} This fact would imply $\rho(n-2)=0$ and $\rho(n-1)=n-1$, or the other way round. But, by definition of $\PP_n$, $\rho(n-1)$ cannot be $0$ or $n-1$. So this case is impossible. 
    \item \emph{Exactly one special vertex belongs to the chain.} Necessarily, it must be the case that $\rho(n-2)=0$ or $\rho(n-2)=n-1$, as the other vertex out of the chain cannot be special, as explained in the previous item.
    In the former case (the other one is analogous) we have that $y_{-1}-y_{0}=0$, hence $y_0=0$. Moreover, $y_0$ must appear also in $y_0-y_1$, and either this vertex belongs to the chain, or it does not.
    If not, we have $\rho(n-1)=1$, which is impossible by hypothesis, since $\rho(n-2)=0$.
    Thus, the vertex belongs to the chain. Additionally, since $y_0$ is zero, vertex $y_0-y_1$ must have sign $\leq 0$, but the other special vertex $(n-2,n-1)$ must be part of the chain too, and it has sign $\geq 0$.
    Again, these two facts together imply that the sign of the chain as well as the sign of the non-special vertices of $H$ are all null, hence $\vc y =\vc 0$.
\end{enumerate}
\medskip

We now prove that \eqref{enu:hardfarkastwo}$\implies$\eqref{enu:hardfarkasone}, by contraposition: we assume that \eqref{enu:hardfarkasone} does not hold (i.e., $\pi\not\in\PP_n$), and show that such a permutation is not representable by $G'_n$.

For $\pi$ not to be in $\PP_n$, one of the following must hold:
\begin{enumerate}[a)]
    \item\label{enu:converseone} $\rho(n-1)=0$;
    \item\label{enu:conversetwo} $\rho(n-1)=1 \land \rho(n-2)=0$;
    \item\label{enu:conversethree} $\rho(n-1)=n-2 \land \rho(n-2)=n-1$;
    \item\label{enu:conversefour} $\rho(n-1)=n-1$.
\end{enumerate}
If~(\ref{enu:converseone} holds (the case~(\ref{enu:conversefour} is symmetrical), then there must be a coefficient vector $\vc a$ such that $\celg{\vc a}{G'_n}(n-1)>\celg{\vc a}{G'_n}(i)$ for all $i\in[n-2]$.
More explicitly, we must have
\begin{align*}
    \celg{\vc a}{G'_n}(n-1) & = a_0 + 2 a_1 + a_2 + \dots + a_{n-2} \\
    %\celg{\vc a}{G'_n}(n-1) & > \celg{\vc a}{G'_n}(n-2) = a_0 + a_1 + \dots + a_{n-2} + a_{n-1} \\
    %\celg{\vc a}{G'_n}(n-1) & >  \celg{\vc a}{G'_n}(n-3) = a_0 + a_1 + \dots + 2 a_{n-2} \\
    %\celg{\vc a}{G'_n}(n-1) & > \celg{\vc a}{G'_n}(i) = a_0 + \dots + a_{i-1} + (n-i-1) a_{i+1} \quad \text{for all $1<i\leq n-2$}
    & > \sum_{j\le i} a_j + (n-i-1) \, a_{i+1} = \celg{\vc a}{G'_n}(i) \quad \text{for all } i\in[n-2].
\end{align*}
This yields the following system of inequalities:
\begin{align*}
    \celg{\vc a}{G'_n}(n-1) & > \celg{\vc a}{G'_n}(n-2) \implies a_1 > a_{n-1} \\
    \celg{\vc a}{G'_n}(n-1) & > \celg{\vc a}{G'_n}(n-3) \implies a_1 > a_{n-2} \\
    \celg{\vc a}{G'_n}(n-1) & > \celg{\vc a}{G'_n}(n-4) \implies a_1 + a_{n-2} > 2 \, a_{n-3} \\
    \celg{\vc a}{G'_n}(n-1) & > \celg{\vc a}{G'_n}(n-5) \implies a_1 + a_{n-3} + a_{n-2} > 3 \, a_{n-4} \\
    \celg{\vc a}{G'_n}(n-1) & > \celg{\vc a}{G'_n}(n-6) \implies a_1 + a_{n-4} + a_{n-3} + a_{n-2} > 4 \, a_{n-5} \\
    & \qquad\qquad\qquad\;\;\, \dots \\
    \celg{\vc a}{G'_n}(n-1) & > \celg{\vc a}{G'_n}(i) \implies a_1 + \sum_{j=i+2}^{n-2} a_j > (n-i-2) \, a_{i+1} \quad\text{for all } i\in[n-4].
    %\celg{\vc a}{G'_n}(n-1) > \celg{\vc a}{G'_n}(n-5) & \implies 3 \, a_1 > a_1 + a_{n-3} + a_{n-2} > 3 \, a_{n-4} \\
    %\celg{\vc a}{G'_n}(n-1) > \celg{\vc a}{G'_n}(n-6) & \implies \qquad 4 a_1 > a_1 + a_{n-4} + a_{n-3} + a_{n-2} > 4 \, a_{n-5} \\
    %& \;\, \dots
    %(n-i-2) \, a_1 > a_1 + \sum_{j=i+1}^{n-2} a_j & > (n-i-2) \, a_{i+1} \quad \text{for all } 0 \le i \le n-4.% \\
    %(n-4) \, a_1 > a_1 + a_3 + a_4 + \dots + a_{n-2} & > (n-4) \, a_{3} i=2 \\ 
    %(n-3) \, a_1 > a_1 + a_3 + \dots + a_{n-2} & > (n-3) \, a_{2} i=1
\end{align*}
Combining each inequality with the previous ones, we obtain:
\begin{equation}\label{eq:converseone}
\begin{split}
    a_1 & > a_{n-1} \\
    a_1 & > a_{n-2} \\
    2a_1 > a_1+a_{n-2}& > 2a_{n-3}  \implies a_1 > a_{n-3} \\
    3a_1 & > 3a_{n-4} \implies a_1 > a_{n-4}\\
    4a_1 & > 4a_{n-5} \implies a_1 > a_{n-5}\\
    & \dots \\
    (n-i-2) \, a_1 & > (n-i-2) \, a_{i+1} \quad \text{for all } 0 \le i \le n-4.
\end{split}
\end{equation}
For $i=0$ this leads to $a_1 > a_1$, which is clearly impossible.
Thus, we conclude that $\pi$ is not representable by $G'_n$ if $\rho(n-1)=0$ (or $\rho(n-1)=n-2$).

%Case~\eqref{enu:conversetwo}, i.e., $\rho(n-1)=1$ and $\rho(n-2)=0$ (as before, case~\eqref{enu:conversethree} is symmetric), implies that $\celg{\vc a}{G'_n}(n-1) > \celg{\vc a}{G'_n}(i)$ for all $i \in [n-3]$.
%Following a similar line of reasoning of~\eqref{enu:converseone}, we obtain a system of inequalities which is basically the same as~\eqref{eq:converseone} without the first inequality, i.e., $a_1 > a_{n-1}$, that in this case should not hold since $\celg{\vc a}{G'_n}(n-2)>\celg{\vc a}{G'_n}(n-1)$.
%Again, we end up with $a_1 > a_1$, which is impossible.

Finally, we are left to prove case~(\ref{enu:conversetwo} (again, case~(\ref{enu:conversethree} is symmetrical). 
It is clear to see that~(\ref{enu:conversetwo} implies that $\celg{\vc a}{G'_n}(n-1)>\celg{\vc a}{G'_n}(i)$ for all $i\in[n-3]$.
As before, this yields the following system of inequalities:
\begin{align*}
    a_1 & > a_{n-2} \\
    a_1 & > a_{n-3} \\
    a_1 & > a_{n-4} \\
    & \dots \\
    (n-i-2) \, a_1 & > (n-i-2) \, a_{i+1} \quad \text{for all } 0 \le i \le n-4.
\end{align*}
Hence, for $i=0$, this leads to the same contradiction $a_1 > a_1$.\qed
\end{proof}

As a consequence:
\begin{corollary}
    Almost all permutations are representable by the graph $G'_n$. More precisely
    \[
        \lim_{n \to \infty} R(G'_n)=1,
    \] 
\end{corollary}

\section{Conclusion and Open Problems}
This paper presented and studied for the first time the class of linear geometric centralities, providing general properties of expressivity and robustness.
Here is a list of problems this paper suggests but that we were unable to solve:
\begin{itemize}
    \item Is convexity necessary for a linear centrality to be rank monotone (Proposition~\ref{prop:rankmon})?
    \item Is it possible to find strongly connected graphs to distinguish pairs of non-proportional coefficients, in the cases of Theorem~\ref{thm:distinguish} where we had to rely on disconnected graphs?% $\widehat{G}_{h,k}$?\todo{slightly adapt this to the new proof}
    \item Is it possible to find some family $G''_n$ of strongly connected graphs such that $R(G''_n)=1$, at least for sufficiently large $n$?
\end{itemize}
A path that might be worth investigating is the study of distance-count matrices, as they vehiculate more information than degree sequences~\cite{BH90}, but have appeared much less in the literature.
For instance, it would be interesting to figure out if there exists a polynomial-time algorithm that decides whether a given matrix is the distance-count matrix of some graph;
as far as we know, this is still an open problem.

{\footnotesize
    \section*{\small Acknowledgements}
    
    This work was supported in part by project SERICS (PE00000014) under the NRRP
    MUR program funded by the EU - NGEU. Views and opinions expressed are however
    those of the authors only and do not necessarily reflect those of the European
    Union or the Italian MUR. Neither the European Union nor the Italian MUR can be
    held responsible for them.
}

\bibliographystyle{splncs04}
\bibliography{biblio,extra}

\hyphenation{ Vi-gna Sa-ba-di-ni Kath-ryn Ker-n-i-ghan Krom-mes Lar-ra-bee Pat-rick Port-able Post-Script Pren-tice Rich-ard Richt-er Ro-bert Sha-mos Spring-er The-o-dore Uz-ga-lis }
\begin{thebibliography}{10}
\providecommand{\url}[1]{\texttt{#1}}
\providecommand{\urlprefix}{URL }
\providecommand{\doi}[1]{https://doi.org/#1}

\bibitem{AntRG}
Anthonisse, J.M.: The rush in a directed graph. Tech. Rep. BN 9/71, Mathematical Centre, Amsterdam (1971)

\bibitem{babai1983canonical}
Babai, L., Luks, E.M.: Canonical labeling of graphs. In: Proceedings of the fifteenth annual ACM symposium on Theory of computing. pp. 171--183 (1983)

\bibitem{BavMMGS}
Bavelas, A.: {A mathematical model for group structures}. Human Organization  \textbf{7},  16--30 (1948)

\bibitem{BavCPTOG}
Bavelas, A.: Communication patterns in task-oriented groups. J. Acoust. Soc. Am.  \textbf{22}(6),  725--730 (1950)

\bibitem{BDFSRSMCBDDC}
Boldi, P., D'Ascenzo, D., Furia, F., Vigna, S.: Score and rank semi-monotonicity for closeness, betweenness, and distance--decay centralities. Social Network Analysis and Mining  \textbf{14}(1), ~183 (Sep 2024)

\bibitem{BFFLGC}
Boldi, P., Furia, F., Prezioso, C.: Linear geometric centralities. In: Bloznelis, M., Drungilas, P., Kami{\'{n}}ski, B., Pra{\l}at, P., {\v{S}}ileikis, M., Th{\'e}berge, F., Vaicekauskas, R. (eds.) Modelling and Mining Networks. pp. 1--16. Springer Nature Switzerland, Cham (2025)

\bibitem{BLVRMCM}
Boldi, P., Luongo, A., Vigna, S.: Rank monotonicity in centrality measures. Network Science  \textbf{5}(4),  529--550 (2017)

\bibitem{BoVWFI}
Boldi, P., Vigna, S.: The {W}eb{G}raph framework {I}: {C}ompression techniques. In: Proc. of the Thirteenth International World Wide Web Conference (WWW 2004). pp. 595--601. ACM Press, Manhattan, USA (2004)

\bibitem{BoVAC}
Boldi, P., Vigna, S.: Axioms for centrality. Internet Math.  \textbf{10}(3-4),  222--262 (2014)

\bibitem{BolMGT}
Bollob{\'a}s, B.: Modern graph theory, Graduate Texts in Mathematics, vol.~184. Springer--Verlag (1998)

\bibitem{BENAMF}
Brandes, U., Erlebach, T.: Network Analysis: Methodological Foundations (Lecture Notes in Computer Science). No.~3418 in Lecture Notes in Computer Science, Springer-Verlag (2005)

\bibitem{BH90}
Buckley, F., Harary, F.: Distance in Graphs. Addison-Wesley Pub. Co.,, Redwood City, Calif. (1990)

\bibitem{farkas1902theorie}
Farkas, J.: Theorie der einfachen ungleichungen. Journal f{\"u}r die reine und angewandte Mathematik (Crelles Journal)  \textbf{1902}(124),  1--27 (1902)

\bibitem{gale1951linear}
Gale, D., Kuhn, H.W., Tucker, A.W.: Linear programming and the theory of games. Activity analysis of production and allocation  \textbf{13},  317--335 (1951)

\bibitem{GarAFCN}
Garg, M.: Axiomatic foundations of centrality in networks. Social Science Research Network  (2009)

\bibitem{KatNSIDSA}
Katz, L.: A new status index derived from sociometric analysis. Psychometrika  \textbf{18}(1),  39--43 (1953)

\bibitem{KisOCFG}
Kishi, G.: On centrality functions of a graph. In: Saito, N., Nishizeki, T. (eds.) Graph Theory and Algorithms. pp. 45--52. Springer Berlin Heidelberg, Berlin, Heidelberg (1981)

\bibitem{KiTTCFDG}
Kishi, G., Takeuchi, M.: A type of centrality functions for a directed graph. Electronics and Communications in Japan (Part I: Communications)  \textbf{66}(5),  38--47 (1983)

\bibitem{LinFSR}
Lin, N.: Foundations of Social Research. McGraw-Hill, New York (1976)

\bibitem{RGRNM}
Roginski, J.W., Gera, R.M., Rye, E.C.: The neighbor matrix: Generalizing the degree distribution (2016)

\bibitem{SabCIG}
Sabidussi, G.: The centrality index of a graph. Psychometrika  \textbf{31}(4),  581--603 (1966)

\bibitem{ScBRCSN}
Schoch, D., Brandes, U.: Re-conceptualizing centrality in social networks. European Journal of Applied Mathematics  \textbf{27}(6),  971--985 (2016)

\bibitem{shafarevich2012linear}
Shafarevich, I.R., Remizov, A.O.: Linear algebra and geometry. Springer Science \& Business Media (2012)

\bibitem{SkSADBC}
Skibski, O., Sosnowska, J.: Axioms for distance-based centralities. Proceedings of the AAAI Conference on Artificial Intelligence  \textbf{32}(1) (Apr 2018)

\end{thebibliography}

\begin{subappendices}

\section*{A. Proof of Theorem~\ref{thm:rhomid}}

For the reader's convenience we restate the theorem:
\begin{theorem-non}
    For $n>3$, consider the graph $G_n$ in Figure~\ref{fig:mdc_example}, and let $\pi \in S_n$ and $\rho=\pi^{-1}$. The following two statements are equivalent:
    \begin{enumerate}
        \item $\rho(0)<\rho(n-1)<\rho(1)$ or $\rho(1)<\rho(n-1)<\rho(0)$;
        \item the system~\eqref{eqn:farkas} has no non-negative non-trivial solution.
    \end{enumerate}
\end{theorem-non}
\begin{proof}
We start proving that (1)$\implies$(2). 
Let $\vc y$ be a non-negative solution to~\eqref{eqn:farkas-ter}; we will prove that necessarily $\vc y = \vc 0$, using Corollary~\ref{cor:farkas}.
The last $n-2$ columns of $C_{G_n}$ can be combined by subtracting from each column all the following ones, and then dividing by a suitable coefficient so that the topmost non-zero element is a $1$, yielding 
{
\begin{equation*} 
    \setlength{\arraycolsep}{3pt}
    \renewcommand\arraystretch{0.9}
    \begin{bmatrix}
    1 & n-1 & 0 & 0 & \dots & 0 & 0 & 0\\
    1 & 1 & 1 & 0 & \dots & 0 & 0 & 0\\
    1 & 1 & 0 & 1 & \dots & 0 & 0 & 0\\
    \vdots & \vdots & \vdots & \vdots & \ddots & \vdots & \vdots & \vdots \\
    1 & 1 & 0 & 0 & \dots & 1 & 0 & 0\\
    1 & 1 & 0 & 0 & \dots & 0 & 1 & 0\\
    1 & 1 & 0 & 0 & \dots & 0 & 0 & 1\\
    1 & 2 & \frac{n-3}{n-2} & 0 & \dots & 0 & 0 & 0\\
    \end{bmatrix}.
\end{equation*}
}
Now doing something similar for the first two columns (we subtract from each of them the last $n-2$ columns, and multiply by a suitable constant), we obtain
{
\begin{equation*} 
    \setlength{\arraycolsep}{3pt}
    \renewcommand\arraystretch{0.9}
    \begin{bmatrix}
    1 & 1 & 0 & 0 & \dots & 0 & 0 & 0\\
    0 & 0 & 1 & 0 & \dots & 0 & 0 & 0\\
    0 & 0 & 0 & 1 & \dots & 0 & 0 & 0\\
    \vdots & \vdots & \vdots & \vdots & \ddots & \vdots & \vdots & \vdots \\
    0 & 0 & 0 & 0 & \dots & 1 & 0 & 0\\
    0 & 0 & 0 & 0 & \dots & 0 & 1 & 0\\
    0 & 0 & 0 & 0 & \dots & 0 & 0 & 1\\
    \frac1{n-2} & \frac1{n-2} & \frac{n-3}{n-2} & 0 & \dots & 0 & 0 & 0\\
    \end{bmatrix}.
\end{equation*}
}
Let $T=\{\rho(i) \mid i=2,\dots,n-2\}$ (i.e., looking at the left-hand side of condition~\eqref{eqn:farkas-ter}, the indices of $\vc y$ involved in the last $n-3$ columns). Let also $p=\rho(0)$, $q=\rho(n-1)$, $r=\rho(1)$ (indices involved in the first three columns): by hypothesis, either $p<q<r$ or $r<q<p$. We assume $p<q<r$ (the other case is symmetrical). Note that necessarily $p<n-1$, $0<r$ and $0<q<n-1$.

\noindent\emph{Case 1.} If $0 \in T$, the condition imposed by the last $n-3$ columns imply
\begin{eqnarray*}
    y_t &=& y_{t-1} \quad \forall t \in T,\ 0<t<n-1\\
    y_0 &=& 0.
\end{eqnarray*}

On the other hand, for the first two (identical) columns condition~\eqref{eqn:farkas-ter} yields:
\[
        y_{p-1}-y_p+\frac1{n-2} \left(y_{q-1}-y_q\right) = 0,
\]
since $p,q\neq0,n-1$. Now start from $k=q-1$: if $k \in T$ then $y_k=y_{k-1}$, so we can iterate with a smaller value of $k$, until eventually $k=p$. Hence $y_{q-1}=y_{q-2}=\dots=y_p$.
Similarly, starting from $k=p-1$, if $k \in T$ we have $y_k=y_{k-1}$, and we can iterate until $k=0$. Hence $y_{p-1}=y_{p-2}=\dots=y_0=0$.

As a consequence, the last equation becomes
\begin{eqnarray*}
        -y_p+\frac1{n-2} \left(y_p-y_q\right) &=& 0\\
        -(n-2) y_p + y_p &=& y_q\\
        -(n-3) y_p &=& y_q.
\end{eqnarray*}
Since both $y_p$ and $y_q$ are non-negative, we must have $y_p=y_q=0$.

The equation implied by the third column is (using Iverson's bracket notation)
\[
    -y_r[r<n-1]+y_{r-1}+\frac{n-3}{n-2} \left(y_{q-1}-y_q\right) = 0,
\]
so plugging in $y_q=y_{q-1}=0$ we distinguish two cases:
\begin{itemize}
    \item if $r=n-1$, then $y_{n-2}=0$;
    \item if, instead $r<n-1$, then $y_r=y_{r-1}$, and starting from $k=n-2\in T$ (hence $y_k=y_{k-1}$), we can iterate with a smaller $k$ until reaching $r$, obtaining $y_{n-2}=y_{n-3}=\dots=y_{r}$.
\end{itemize}
In both cases, we can apply the same reasoning, starting from $k=r-1$, down to $q$, yielding $y_{r-1}=y_{r-2}=\dots=y_q=0$.

Finally, we notice that in all cases we have a set of equations involving all indices $t$, and having the form either $y_t=y_{t-1}$ or $y_t=0$ (and in particular $y_0=0$). So $\vc y=\vc 0$.

\noindent{\emph Case 2.} If $0 \not\in T$, then necessarily $p=0$. The equation for the first two columns this time is    
\[
        -y_0+\frac1{n-2} \left(y_{q-1}-y_q\right) = 0,
\]
and again starting from $k=q-1$ and iterating we obtain $y_{q-1}=\dots=y_0$, yielding
\[
        -(n-2)y_0 + y_0=y_q.
\]
Once more $y_0=y_q=0$, and we can proceed as above.

\medskip
We are now going to prove that (2)$\implies$(1), by contraposition.    
Let us assume that (1) does not hold, thus $\rho(n-1)>\rho(0)$ and $\rho(n-1)>\rho(1)$ (or, symmetrically, $\rho(0)>\rho(n-1)$ and $\rho(1)>\rho(n-1)$).

Under such conditions, given a coefficient vector $\vc a$, we have
\begin{align*}
    & \celg{\vc a}{G_n}(0) = a_0 + (n-1) a_1 \\
    & \celg{\vc a}{G_n}(1) = a_0 + a_1 + (n-2) a_2 \\
    & \celg{\vc a}{G_n}(n-1) = a_0 + 2 a_1 + (n-3) a_2,
\end{align*}
and, by hypothesis,
\begin{align*}
    a_0 + 2 a_1 + (n-3) a_2 & > a_0 + (n-1) a_1 && \implies a_2 > a_1 \\
    a_0 + 2 a_1 + (n-3) a_2 & > a_0 + a_1 + (n-2) a_2 && \implies a_1 > a_2,
\end{align*}
which is impossible.
Thus, $\pi$ is not representable by $G_n$, and by Theorem~\ref{thm:reprfark} system~\eqref{eqn:farkas} has a nonnegative non trivial solution $\vc y$, i.e., (2) cannot hold.\qed
\end{proof}
\end{subappendices}
\end{document}